\documentclass[sigconf]{acmart}

\AtBeginDocument{%
  \providecommand\BibTeX{{%
    \normalfont B\kern-0.5em{\scshape i\kern-0.25em b}\kern-0.8em\TeX}}}
    
\usepackage[utf8]{inputenc}
\usepackage{graphicx}
\usepackage{wrapfig}
\usepackage{floatflt}
\usepackage{subcaption}
\usepackage{mdframed}
\usepackage{tcolorbox}
\usepackage{mathtools}

\usepackage{fancyhdr}
\newcommand{\protocol}{{\sf Advocate} }
\newcommand{\protocolnosp}{{\sf Advocate}\ignorespaces}

\newcommand{\mc}{\mathcal}

\newcommand{\eg}{e.g.~}
\newcommand{\ie}{i.e.,~}
\newcommand{\st}{s.t.~}
\newcommand{\wrt}{w.r.t.~}

\newcommand{\party}{\mc{P}}
\newcommand{\adversary}{\mc{A}}

\newcommand{\persistenceParam}{k}
\newcommand{\livenessParam}{u}
\newcommand{\miningpower}{\mu}

\newcommand{\chain}{\mc{C}}

\usepackage{xparse}
\newcommand{\bnm}{\begin{newmath}}
\newcommand{\enm}{\end{newmath}}

\newcommand{\bea}{\begin{eqnarray*}}%
\newcommand{\eea}{\end{eqnarray*}}%

\newcommand{\bne}{\begin{newequation}}
\newcommand{\ene}{\end{newequation}}

\newcommand{\bal}{\begin{newalign}}
\newcommand{\eal}{\end{newalign}}

\newenvironment{newalign}{\begin{align}%
\setlength{\abovedisplayskip}{4pt}%
\setlength{\belowdisplayskip}{4pt}%
\setlength{\abovedisplayshortskip}{6pt}%
\setlength{\belowdisplayshortskip}{6pt} }{\end{align}}

\newenvironment{newmath}{\begin{displaymath}%
\setlength{\abovedisplayskip}{4pt}%
\setlength{\belowdisplayskip}{4pt}%
\setlength{\abovedisplayshortskip}{6pt}%
\setlength{\belowdisplayshortskip}{6pt} }{\end{displaymath}}

\newenvironment{newequation}{\begin{equation}%
\setlength{\abovedisplayskip}{4pt}%
\setlength{\belowdisplayskip}{4pt}%
\setlength{\abovedisplayshortskip}{6pt}%
\setlength{\belowdisplayshortskip}{6pt} }{\end{equation}}

\newcounter{ctr}

\newcounter{mytable}
\def\mytable{\begin{centering}\refstepcounter{mytable}}
\def\endmytable{\end{centering}}

\newcounter{myfig}
\def\myfig{\begin{centering}\refstepcounter{myfig}}
\def\endmyfig{\end{centering}}

\newlength{\saveparindent}
\newlength{\saveparskip}
\newcommand{\E}{{\rm I\kern-.3em E}}

\renewcommand{\eqref}[1]{\mbox{Equation~(\ref{#1})}}

\def \part {part}

\renewcommand{\paragraph}[1]{\vspace*{6pt}\noindent\textbf{#1}\;}

\def \blackslug{\hbox{\hskip 1pt \vrule width 4pt height 8pt
    depth 1.5pt \hskip 1pt}}
\def \qed{\quad\blackslug\lower 8.5pt\null\par}

\newcounter{mynote}[section]

\newcommand\ignore[1]{}

\newcounter{rcnote}[section]

\newcounter{mrnote}[section]

\newcounter{fknote}[section]

\newcounter{anote}[section]

\DeclareMathSymbol{\mlq}{\mathord}{operators}{``}
\DeclareMathSymbol{\mrq}{\mathord}{operators}{`'}

\newcommand{\rhf}[2]{R_{f, \gamma}}

\DeclareDocumentCommand{\edist}{o o}{
  \ensuremath{
    \IfNoValueTF{#1}{{d}}{{\sf d}(#1,#2)}
  }
}

\newcommand{\olrk}[1]{\ifx\nursymbol#1\else\!\!\mskip4.5mu plus 0.5mu\left(\mskip0.5mu plus0.5mu #1\mskip1.5mu plus0.5mu \right)\fi}

\NewDocumentCommand{\indseq}{ O{1} O{r} }{{#1}\ldots {#2}}

\usepackage{algorithm}
\usepackage{xcolor}
\definecolor{azure}{rgb}{0.54, 0.17, 0.89}

\usepackage{enumitem,kantlipsum}

\newcommand{\randomnessnosp}{Stochastic\ignorespaces}

\newcommand{\basicnosp}{Nakamoto\ignorespaces}

\usepackage{mathtools}
\usepackage{hyperref}

\newcommand{\calR}{{\mathcal R}}

\usepackage{tablefootnote}

\copyrightyear{2022}
\acmYear{2022}
\setcopyright{acmlicensed}\acmConference[MobiHoc '22]{The Twenty-third International Symposium on
Theory, Algorithmic Foundations, and Protocol Design for Mobile Networks
and Mobile Computing}{October 17--20, 2022}{Seoul, Republic of Korea}
\acmBooktitle{The Twenty-third International Symposium on Theory,
Algorithmic Foundations, and Protocol Design for Mobile Networks and
Mobile Computing (MobiHoc '22), October 17--20, 2022, Seoul, Republic of
Korea}
\acmPrice{15.00}
\acmDOI{10.1145/3492866.3549731}
\acmISBN{978-1-4503-9165-8/22/10}

\begin{document}
\title{ Optimal Bootstrapping of  PoW Blockchains}
%
%

\author{Ranvir Rana}
\affiliation{%
  \institution{University of Illinois Urbana Champaign}
  \country{USA}}
\email{rbrana2@illinois.edu}

\author{Dimitris Karakostas}
\affiliation{
  \institution{University of Edinburgh}
  \country{UK}}
\email{d.karakostas@ed.ac.uk}

\author{Sreeram Kannan}
\affiliation{
  \institution{University of Washington Seattle}
  \country{USA}}
\email{ksreeram@uw.edu}

\author{Aggelos Kiayias}
\affiliation{
  \institution{University of Edinburgh and IOHK}
  \country{UK}}
\email{akiayias@inf.ed.ac.uk}

\author{Pramod Viswanath}
\affiliation{
  \institution{University of Illinois Urbana Champaign}
  \country{USA}}
\email{pramodv@illinois.edu}

%
%
%
\begin{abstract}
     Proof of Work (PoW) blockchains are susceptible to adversarial majority mining attacks in the early stages due to incipient participation and corresponding low net hash power.  Bootstrapping ensures safety and liveness during the {\em transient} stage by protecting  against a majority mining attack, allowing a PoW chain to  grow the participation base and corresponding mining hash power. Liveness is especially important since a loss of liveness will lead to loss of honest mining rewards, decreasing honest participation, hence creating an undesired spiral; indeed existing bootstrapping mechanisms offer especially weak liveness guarantees. 
    
     
    In this paper, we propose \protocolnosp, a new bootstrapping  methodology, which achieves two main results: (a) optimal liveness and low latency under a super-majority adversary for the Nakamoto longest chain protocol and (b) immediate black-box generalization to a variety of parallel-chain based scaling architectures, including {\sf OHIE}  \cite{yu2019ohie} and {\sf Prism} \cite{BagariaKTFV19}. We demonstrate via a full-stack implementation the robustness of \protocol under a 90\% adversarial majority.
    
    \keywords{Bootstrapping  \and PoW \and Checkpointing}
\end{abstract}

\ccsdesc[500]{Security and privacy~Distributed systems security}
\ccsdesc[500]{Computer systems organization~Dependable and fault-tolerant systems and networks}
\ccsdesc[500]{Theory of computation~Distributed algorithms}

\maketitle              
\section{Introduction}



{\em Bootstrapping PoW Blockchains}. Proof of Work (PoW) blockchains, epitomized by {\sf Bitcoin}, have proven themselves to be secure designs. Their security has been shown in theory \cite{EC:GarKiaLeo15} as well as in practice ({\sf Bitcoin} has not seen any severe safety or liveness incidents in more than 13 years of being online). However, an important, and less studied, aspect of PoW blockchains is that they are particularly hard to bootstrap. At the early stages of a PoW blockchain, there is not much participation (in terms of mining hash power), making it relatively easy for an adversary to overpower the honest miners.
If a PoW blockchain can avoid the dangers of such attacks in its infancy, eventually a significant amount of honest hashing power participates in the mining process and security is correspondingly strengthened. Thus, the focus of this paper is on principled approaches to secure {\em bootstrapping}, a crucial aspect to the successful development of any PoW blockchain.

\subsection{Related Works}

{\em Checkpointing}. Checkpointing is a standard technique used in state machine replication
protocols~\cite{castro1999practical}, where a centralized server issues checkpoints attesting to the recent state of the protocol. In  blockchains, checkpoints attest to the hash of well-embedded blocks every so often so that new users can securely bootstrap using  a recent execution state of the protocol. A key benefit of such checkpointing is that an adversary even with super-majority mining power cannot create a long-range reversion attack (i.e., forking from a block created long ago). Since such a long-range fork will deviate from the stated checkpoint, clients will reject them. Practical blockchains utilizing checkpointing include 
Bitcoin~\cite{nakamoto2008bitcoin,btc-checkpoints}, 
Peercoin~\cite{king2012ppcoin}, Feathercoin~\cite{feathercoin-checkpoints}, and
RSK~\cite{rsk}. A centralized checkpointing mechanism was maintained by Satoshi Nakamoto themself (presumably honest) until late 2014. Additionally, 
checkpoints of some form have been central in the context of Proof-of-Stake (PoS) protocols, including 
Ouroboros~\cite{C:KRDO17}, Snow White~\cite{FC:DaiPasShi19}, and Ouroboros
Praos~\cite{EC:DGKR18}, as well as \eg in hybrid
consensus~\cite{DBLP:conf/wdag/PassS17}, Thunderella~\cite{DBLP:conf/eurocrypt/PassS18},
ByzCoin~\cite{DBLP:conf/uss/Kokoris-KogiasJ16}, and Algorand~\cite{DBLP:conf/sosp/GiladHMVZ17}.
In a related context, 
Fantomette~\cite{azouvi18} employs distributed checkpoints to secure a
blockDAG-based ledger.  


{\em Finality Gadgets}: 
A crucial problem in the context of checkpointing is ensuring the safe delivery of checkpoints to the blockchain client.
A new generation of blockchain algorithms have sought to decentralize this process by designing a separate distributed consensus protocol to issue checkpoints. We will refer to this class of solutions as finality gadgets, which are comprised of a Byzantine Fault Tolerant (BFT) protocol for finalizing blocks created  by a Proof-of-work (PoW) or Proof-of-stake (PoS) chain protocol. They have become very popular methods for combining the best features of the BFT and PoW protocols and are proposed for deployment in many major blockchains including Ethereum 2.0 \cite{eth2,DBLP:journals/corr/abs-2007-01560}. Depending on the context, the checkpoint committee can be comprised of a fixed committee (for example, run by independent community leaders) or the committee itself can be elected using stake deposits. 

{\em Rationale for Finality Gadgets}. 
There are multiple reasons for utilizing finality gadgets, and different protocols emphasize different properties. We enumerate the properties that motivate the development of finality gadgets as follows. 
\begin{enumerate}
    \item Safety against long-range attacks;
   \item Economic finality;
    \item Availability vs Finality tradeoff capability; 
    \item Responsiveness: low-latency block confirmation; 
\end{enumerate}

The {\em raison d’etre} of finality-gadgets (including centralized checkpoints) is to provide property (1):  safety against long-range attacks by an adversarial majority. In a PoW system, this is an important safeguard 
since an adversary can control a super-majority mining power temporarily, for example, by renting cloud mining equipment. 
Finality gadgets and checkpoints can prevent such attacks. 

Beyond this basic reason, different protocols optimize for different criterion. 
Casper~\cite{buterin2017casper,casper-incentives} focuses on (2) that ensures that if safety is violated, the malicious action is detected and at least one-third of the staking nodes will lose their stake, a similar approach is taken by GRANDPA~\cite{DBLP:journals/corr/abs-2007-01560} and 
Afgjort~\cite{DBLP:conf/scn/Dinsdale-YoungM20}. Recent works \cite{polygraph,BFTforensics} have identified general BFT protocols which have such detectability.
Some protocols like \cite{ebbandflow,sankagiri2020checkpointed} focus on (3) by designing gadgets that let users prioritize adaptivity or availability by implementing different confirmation rules. 
Some protocols like Winkle \cite{winkle} guarantee (1) by utilizing transaction traffic for voting. Finally, other finality gadgets (\eg Afgjort~\cite{DBLP:conf/scn/Dinsdale-YoungM20}) are optimized for property (4): responsiveness, the ability of the BFT to confirm blocks produced by the PoW chain near-instantaneously. 

{\em Bootstrapping gadget:} A bootstrapping gadget is a finality gadget that has one additional property: {\bf Liveness despite adversarial majority}. During the initial stages of a novel PoW protocol, liveness is required to ensure that honest miners are rewarded for their effort even under an adversarial majority. Lack of a live checkpointing protocol creates an undesired spiral: low honest participation $\rightarrow$ low honest miner rewards $\rightarrow$ honest miners leaving (lower participation).

\subsection{Motivation and contributions}

{\em Key shortcoming of existing solutions}. While all the aforementioned protocols satisfy several properties of finality gadgets, none of them can work as bootstrapping gadgets. This was observed in a recent paper \cite{cryptoeprint:2020:173} for even the simple centralized checkpointing protocol. An adversary controlling a majority mining power can issue a long private adversarial chain, which does not contain any honest transactions. When such a chain is checkpointed repeatedly, honest clients lose liveness in the system. The paper proposed the inclusion of a random nonce as well as a checkpoint certificate (issued by the central checkpointer or BFT) to reduce the impact of this attack. The key idea is that since this random nonce needs to be included in the next block, this creates a renewal event where the adversarial blocks stored prior to the event have to be disregarded, creating a new race between the honest and adversarial chains at each checkpoint. While this approach can ensure a non-zero chance that the honest chain can win, thus giving asymptotic liveness, the latency of transaction inclusion as well as the chain quality (fraction of honest blocks in the final ledger) and the corresponding mining rewards for honest miners decrease exponentially as mining power increases beyond $50 \%$ or as the inter-checkpoint interval increases.

{\em Main Contributions}. In this paper, we focus on building a bootstrapping gadget that achieves safety {\em and} liveness under an adversarial majority. We propose \protocolnosp, a new scheme that achieves optimal chain quality.
The core idea is the inclusion of appropriate reference links to checkpoint blocks. Variations of this idea have been proposed in different contexts in the literature: in \cite{EC:GarKiaLeo15} for achieving a $1/2$ threshold Byzantine Agreement, in Fruitchains \cite{fruitchains} for designing incentives, in inclusive protocols \cite{inclusive} for minimizing block wastage due to forking and in general DAG (directed acyclic graph) protocols, such as Conflux \cite{conflux}, for improving throughput. We prove that \protocol achieves optimal chain quality while ensuring transaction inclusion for all honest transactions within two epochs even under an arbitrarily high adversarial mining power (the so-called ``99\% mining adversary''). The plots for chain quality of related works are in Figure \ref{fig:cq_comparison}; an upper bound on the chain quality of \cite{cryptoeprint:2020:173} diminishes rapidly with adversarial power ($\beta$) and epoch size ($e$), whereas \protocol achieves the {\em optimal chain quality} equal to the honest mining power ($1-\beta$). 


\begin{figure}
    \centering
    \includegraphics[width = .35\textwidth]{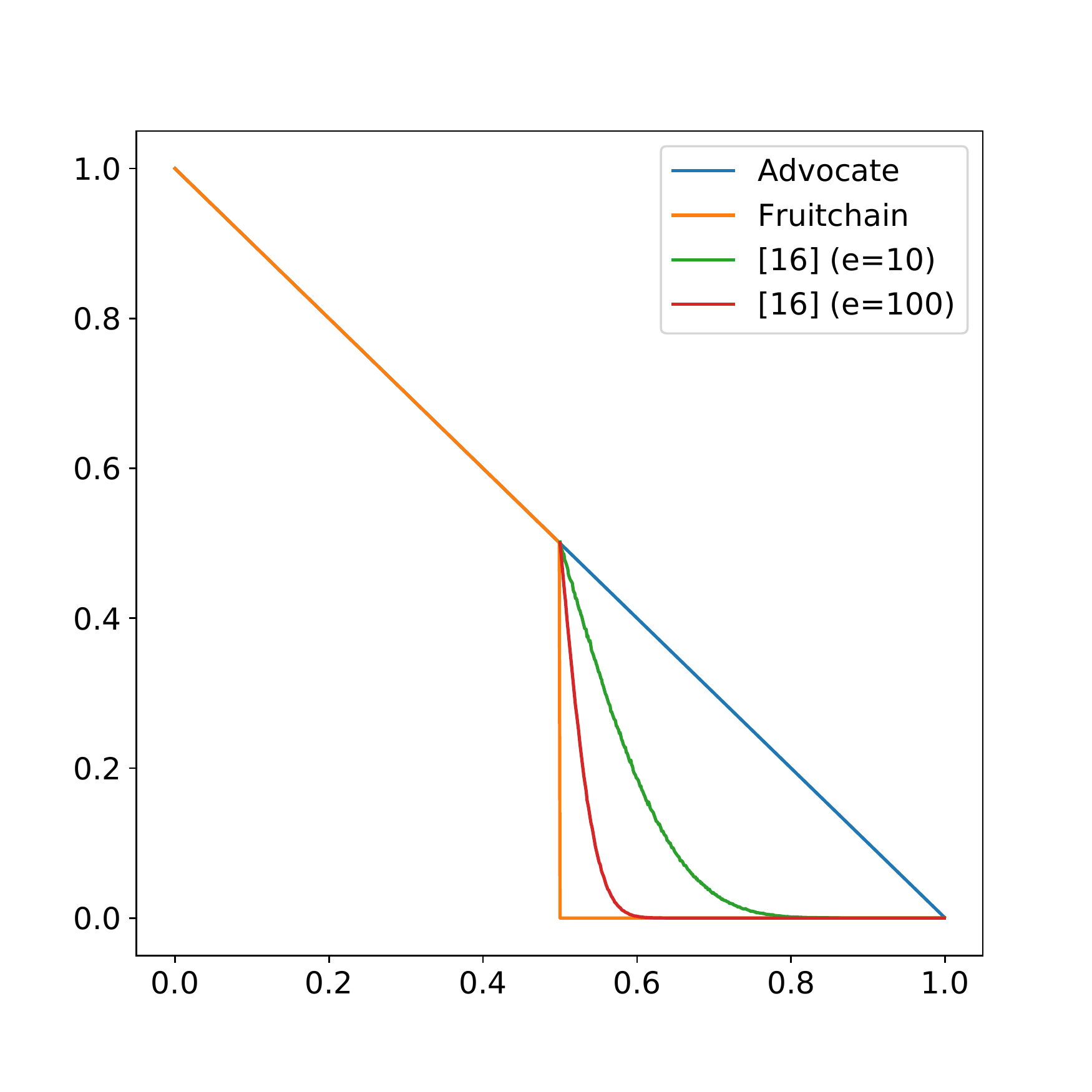}
    \put(-175,60){\rotatebox{90}{Chain quality}}
    \put(-140,5){Adversary mining power: $\beta$}
    \put(-100,110){$\beta=0.5$}
    \caption{Chain quality (CQ) of fruitchains deteriorates to 0 when $\beta>0.5$; an upper bound on CQ of previous work \cite{cryptoeprint:2020:173} deteriorates rapidly with epoch length and $\beta$; CQ of \protocol is optimal.}
    \label{fig:cq_comparison}
\end{figure}

{\em General applicability}. We create appropriate blackbox interfaces through which our protocol can employ any BFT protocol for checkpointing, to make our construction widely applicable. We also demonstrate that our bootstrapping gadget \protocol work with a variety of PoW protocols beyond Nakamoto consensus. 

{\em System Implementation}. We perform extensive experiments on a distributed testbed to demonstrate the robustness of our protocol under and up to  90\% adversarial mining majority and compare performance gains with prior state-of-the-art. To demonstrate compatibility of \protocol with high throughput parallel-chain architectures, we implemented \protocol on {\sf Prism} and demonstrated an honest throughput of 8,200 tx/s under a  70\% adversarial majority. The code for the systems implementation is available at \cite{advocate_sys_anon}.

{\bf Organization}. The rest of the paper is organized as follows. Section \ref{sec:preliminaries} provides an overview of the preliminaries used in our work including the threat model and the distributed ledger’s properties and block production mechanisms. Section \ref{sec:ideal} describes \protocol under a single checkpointing node and provides its security analysis to show safety and liveness with optimal chain quality. Section \ref{sec:bft} extends \protocol to ensure similar performance and security guarantees under a committee-based BFT-SMR protocol by providing a unified network functionality. Section \ref{sec:prism} integrates \protocol into parallel-chain architectures by providing a {\em meta-protocol} that can be readily integrated into {\sf Prism}, {\sf OHIE} and ledger-combiner to achieve high throughput under an adversarial majority. 
Section \ref{sec:evaluation} presents detailed experimental results on a distributed testbed of a full-stack implementation, stress testing performance with majority adversarial mining power.

\section{Preliminaries}\label{sec:preliminaries}


\subsection{The Distributed Ledger Model}\label{subsec:ledger-model}

The distributed ledgers analyzed in this work are constructed as blockchains. A
ledger is formed as a hash chain (or tree) of blocks, each block containing
transactions which alter the ledger's state. New blocks, which extend the
chain, are created by mining parties at regular intervals. Conflicts, \ie forks in the
chain, are resolved following Sybil resilience mechanisms, such as
PoW. Given a tree of blocks, each party chooses a 
single branch as the \emph{main chain}; blocks that are not part of the main 
chain are called \emph{uncles}.

We assume a synchronous setting with a delay upper bound of $\Delta$. Specifically, the execution proceeds in rounds. On each round, every party is
activated to participate in the protocol. Communication is performed via a
``diffuse'' functionality, \ie a gossip protocol, such that no point-to-point
communication channels exist, but rather a peer-to-peer network is formed.
Therefore, every message produced at round $r$ is received by all other parties
by round $r + 1$. We also assume that the number of participating parties is
fixed for the duration of the execution.

The ledger's core properties, described in detail by the Bitcoin Backbone
model~\cite{EC:GarKiaLeo15}, are provided in Definitions~\ref{def:stable}-\ref{def:liveness}.

\begin{definition}[Stable Block and Transaction]\label{def:stable}
    A block is \emph{stable} if it is $\persistenceParam$-deep in the main chain.
    A transaction published in a stable block is also stable.
\end{definition}

\begin{definition}[Safety]\label{def:safety}
    A transaction reported as \emph{stable} by an honest party on round $r$ is
    reported as stable by all honest parties on round $r+1$, at the same position in the
    ledger.
\end{definition}

\begin{definition}[Liveness]\label{def:liveness}
    A transaction which is provided continuously as input to the parties
    becomes stable after $\livenessParam$ rounds.
\end{definition}

An additional important property  of interest is {\em chain quality} 
(Definition~\ref{def:chain-quality})~\cite{EC:GarKiaLeo15,LC:KiaPan17}.   
Briefly, this property ensures that the number of blocks that each party contributes 
to the chain is bounded by a function of the party's mining power $\miningpower$.

\begin{definition}[Chain Quality $(q,l)$]\label{def:chain-quality}
    Let $q$ be the proportional mining power of $\adversary$.
    Chain quality with parameter $l$ states that for any honest party $\party$ with 
    chain $\chain$, it holds that, for any $l$ consecutive blocks of $\chain$, the 
    ratio of adversarial blocks is at most $1-q$.
\end{definition}

\subsection{Threat Model}\label{subsec:threat-model}

Our work considers polynomial-time executions, such that all parties, including
the adversary $\adversary$, are locally polynomial-bounded. On each execution
round, the adversary may ``corrupt'' a party, at which point it accesses the
party's internal state; following, when the corrupted party is supposed to be
activated, the adversary is activated instead. Additionally, $\adversary$ is
``adaptive'', \ie corrupts parties on the fly, and ``rushing'', \ie retrieves
all honest parties' messages before deciding its strategy at each round.

$\adversary$ controls $\mu_\adversary$ of the network's mining power and tries to break safety and liveness. 
To break safety, $\adversary$
forces two non-corrupted nodes to accept different chains as stable, \ie to
report different transactions as stable in the same position in their
respective ledger. To break liveness, $\adversary$ attempts to prevent a
transaction from becoming stable within $u$ rounds.
We explore settings where the honest majority assumption is violated, \ie when
the adversary may control more than $1/2$ of the net mining power. In those
settings, the ledger cannot be secure in a standalone fashion, hence the
need for the checkpointing protocols presented in this work.
\section{\protocolnosp: Optimal Checkpointing of  Longest Chain Protocols}
 \label{sec:ideal}
 Our main contribution is a novel protocol, \protocolnosp, that ensures both safety and liveness against a (arbitrarily high)  super-majority mining adversary on the PoW chain. This section considers a single (honest) checkpointing node, in order to clearly present the main innovations, while the following sections relax this assumption by proposing a distributed checkpointing federation. 
 
 

Checkpointing in \protocol is performed via \emph{certificates}. Specifically, at regular intervals of $e$ blocks on the main chain, the checkpointing service issues a signed certificate, which is published on the chain within $c$ blocks on the main chain (PoW chain). The certificate defines the canonical chain that parties should adopt. \protocol is parameterized by two values $c$ and $e$ as described above.

\subsection{Checkpointing Party Behavior}\label{subsec:single-node-behavior}
The checkpointing party is connected to the blockchain network, so at each round $t$ it holds a view of the PoW chain. Therefore, on any round, the checkpointing party maintains a list of leaves $\mathcal{L}(t)$ of its local block-tree. 

At regular intervals (\ie every $e$ blocks), the party issues a checkpoint certificate. The $i$-th checkpoint certificate issued by the party is denoted $C_i$. 
A checkpoint certificate is constructed as follows: $C_i = \{B_i, \mathcal{R}_i, S_i\}$; $B_i$ is the checkpointed block, \ie the block of the main PoW chain that the party checkpoints; $\mathcal{R}_i$ is a list of references of blocks that are not part of the main chain, \ie leaves of the block tree which are not checkpointed; $S_i$ is the signature of the certificate. The initial, bootstrapping certificate is $C_0 = \{0, \{0\}, S_0\}$. For the rest of the paper, a \emph{checkpointed} block is a block which is referenced in a checkpointing certificate.

 


\subsection{Main Chain behavior}
 
With the introduction of checkpoints, the PoW node behavior needs to change appropriately. The {\em key change} is that, once a new certificate checkpoints block $B_i$, it should be published in at least one of the $c$ blocks that immediately extend $B_i$; the first block that includes the certificate is called the \emph{referring block}. The nodes follow the \emph{longest checkpointed} chain. In summary, \protocol modifies the main-chain rule as follows: 
\begin{itemize}
    \item Go to $B_i$ in the blocktree. 
    \item If there exists a descendant block $B_i^r$ within $c$ blocks of $B_i$ that contains $C_i$, pick the longest chain which contains $B_i^r$ as the main chain. \emph{(Note: A block $B_i^r$ which contains $C_i$ but is more than $c$ blocks after $B_i$ is not acceptable.)} 
   \item If no such $B_i^r$ exists, then either of the following holds:
   \begin{enumerate}
       \item $B_i$ is \emph{not} $c$-deep in the longest chain: pick the longest chain containing $B_i$ as the main-chain.  
       \item $B_i$ is $c$-deep in the longest chain: pick one of the chains which is $(c-1)$-deep and contains $B_i$ as the main chain (breaking ties arbitrarily).
   \end{enumerate}
\end{itemize}
 
\paragraph{Mining behavior}\label{subsec:advocate-miner}
Miners follow the above main chain rules. Additionally, \wrt a checkpoint certificate $C_i$, two cases exist: 
\begin{enumerate}
    \item the main chain contains $C_i$ in some block $B_i^r$: proceed mining as usual.
    \item the main chain does not contain $C_i$: include $C_i$ alongside the list of transaction to be mined.
\end{enumerate}

We suppose that when a miner creates a new block, the block contains all transactions in the miner's mempool (in practice, this requires sufficiently incentivized transactions fees). 
 With hindsight, this assumption will prove useful to argue that a transaction is published in the first honestly-generated block that is produced after the transaction's creation.
 
Let main chain oracle $F_{mco}$ represent the view in the execution of the underlying consensus algorithm (Nakamoto) and  the public global tree $G_t$ at time $t$ is received by time $t+\Delta$ (by the synchronous network assumption). The main chain oracle gets the additional checkpointing information from  \protocol. The interaction of a main chain oracle ($F_{mco}$) with the \protocol functionality can be formalized using a functionality $F_{\protocolnosp}$ described next. 

\vspace{0.3cm}
\begin{mdframed}[backgroundcolor=black!10,rightline=false,leftline=false]

\begin{center}
    {\bf Advocate $F_{Advocate}$}
\end{center}

    \vspace{0.1cm}

$F_{\protocol}$ and $F_{mco}$ interact using various push-pull messages as described bellow; message delay is considered zero since both the modules use the same machine
  \begin{enumerate}
        \item {\em potentialCandidate:} $F_{mco}$ sends this message as soon as it receives a new block $B_n$ that satisfies the checkpoint criteria
        \item {\em candidateFinalized:} $F_{\protocol}$ sends this message as soon as it receives a new checkpointing certificate $\mathcal{C}_n$
        \item {\em sendReferences:} On receiving a 
        a potential candidate, $F_{\protocol}$ immediately checks with the $F_{mco}$ to see if there are any unreferred uncle blocks; $F_{mco}$ replies immediately with {\em unreferredBlocks}
        \item {\em isCertValid:} On receiving a checkpoint certificate $\mathcal{C}_n$ on the main chain, $F_{mco}$ immediately requests $F_{\protocol}$ to check if it's valid (correct signatures, etc.); $F_{\protocol}$ immediately replies in boolean using {\em certValidity}.
        \item {\em isValidBlock}: Triggered when $F_{\protocol}$ receives a new checkpoint certificate; it sends the above message with block hashes to the main chain oracle, to see if those blocks are valid. The main chain oracle replies with {\em validBlock}
        \item {\em validBlock:} Response to the above; if $F_{mco}$ has not received the block yet, it waits for $\Delta$ before the reply. If the block is not yet received or is invalid due to main chain consensus protocol, it replies {\em False}, else it replies {\em True}.
        \item {\em certRequest} A query to the certificate database stored by $F_{\protocolnosp}$, the response {\em requestedCert} is immediate.
    \end{enumerate}

\end{mdframed}
\vspace{0.3cm}

Note that all the interactions between $F_{advocate}$ and $F_{mco}$ are immediate except for {\em isValidBlock} which has a maximum delay of $\Delta$; the synchronous network delay bound. 

The checkpointing party maintains the same functionalities $F_{mco}$ and $F_{Advocate}$ with an additional checkpointing service functionality $F_{cps}$.

\vspace{0.3cm}
\begin{mdframed}[backgroundcolor=black!10,rightline=false,leftline=false]

\begin{center}
    {\bf Checkpointing module $F_{cps}$}
\end{center}

    \vspace{0.1cm}
    $F_{cps}$ receives the {\em checkpointCandidate} and {\em unreferredBlocks} in the form of {\em inputValue} message from $F_{Advocate}$. This message takes zero delay once $F_{Advocate}$ receives {\em potentialCandidate} from $F_{mco}$. $F_{cps}$ queries $F_{Advocate}$ regarding the validity of this input using the {\em isInputValid} and {\em inputValidity} messages and gets a reply after a delay $t_{cps}$. If the input is valid, $F_{cps}$ immediately certifies the input and sends the message to $F_{Advocate}$.
\end{mdframed}
\vspace{0.3cm}

Note that the above process will happen within a time $t_{cps}$ with $t_{cps}/\Delta \ll  1$, since the process is in the same machine. The message {\em isInputValid} and {\em inputValidity} seems redundant for now, however, we will see in section \ref{sec:bft}, that this message classification is critical.  

\subsection{Decoupled Validity: Ledger Creation}\label{subsec:decoupled-validity}

Without loss of generality, we assume that the execution completes with the issuing of a final checkpoint. To construct the aggregate ledger at any point of the execution, the blocks of the main chain are concatenated with the blocks of the non-main branches. In the aggregate ordering, the main chain blocks, up to and including the referring block (\ie the block which includes the checkpointing certificate), are prioritized over the blocks which are not part of the main chain (\ie the ``uncle'' blocks). 

Formally, let $T(\mathcal{L}_i)$ be the tree corresponding to the leaves $\mathcal{L}_i$. Let the forest $F_i$ be the difference between $T(\mathcal{L}_i)$ and the previously checkpointed tree: $F_i := T(\mathcal{L}_i) \setminus \{ T(\mathcal{L}_{i-1}) \cup \text{Chain}_i \}$, where $\text{Chain}_i$ is the main chain up to (and including) the referring block for checkpoint $C_i$. $\pi(\cdot)$ denotes a topological sort of the blocks in a forest, with ties broken in a universal manner (\eg via block hashes). The aggregate ledger is constructed by concatenating the referring block (\ie $\text{Chain}_i$) with $\pi(F_i)$. Therefore, when the referring block for certificate $C_i$ is read, the blocks referred by $C_i$ (\ie the blocks in $\pi(F_i)$) are also read in the order defined in $C_i$.\footnote{We assume that the (honest) checkpointing node follows the universal topological sorting $\pi$ when constructing the reference list of certificate $C_i$.} This procedure is exemplified in figure \ref{fig:protocol} where Blocks 3,5 are included after block 8. Note that we can follow other universal ordering approaches for blocks in $F_i$ (e.g. sort by hash) without affecting the security of our protocol. 

Figure~\ref{fig:protocol} depicts the ledger construction, \st the sanitized ledger (figure \ref{fig:sanitization}) is obtained by parsing the main chain and referenced blocks and removing invalid (\eg double spending) entries. Such post ordering ledger sanitization is required even under a honest setting since referred uncle blocks may have transactions that conflict with transactions on the main chain. 


\begin{figure}
    \centering
    \includegraphics[width = \columnwidth]{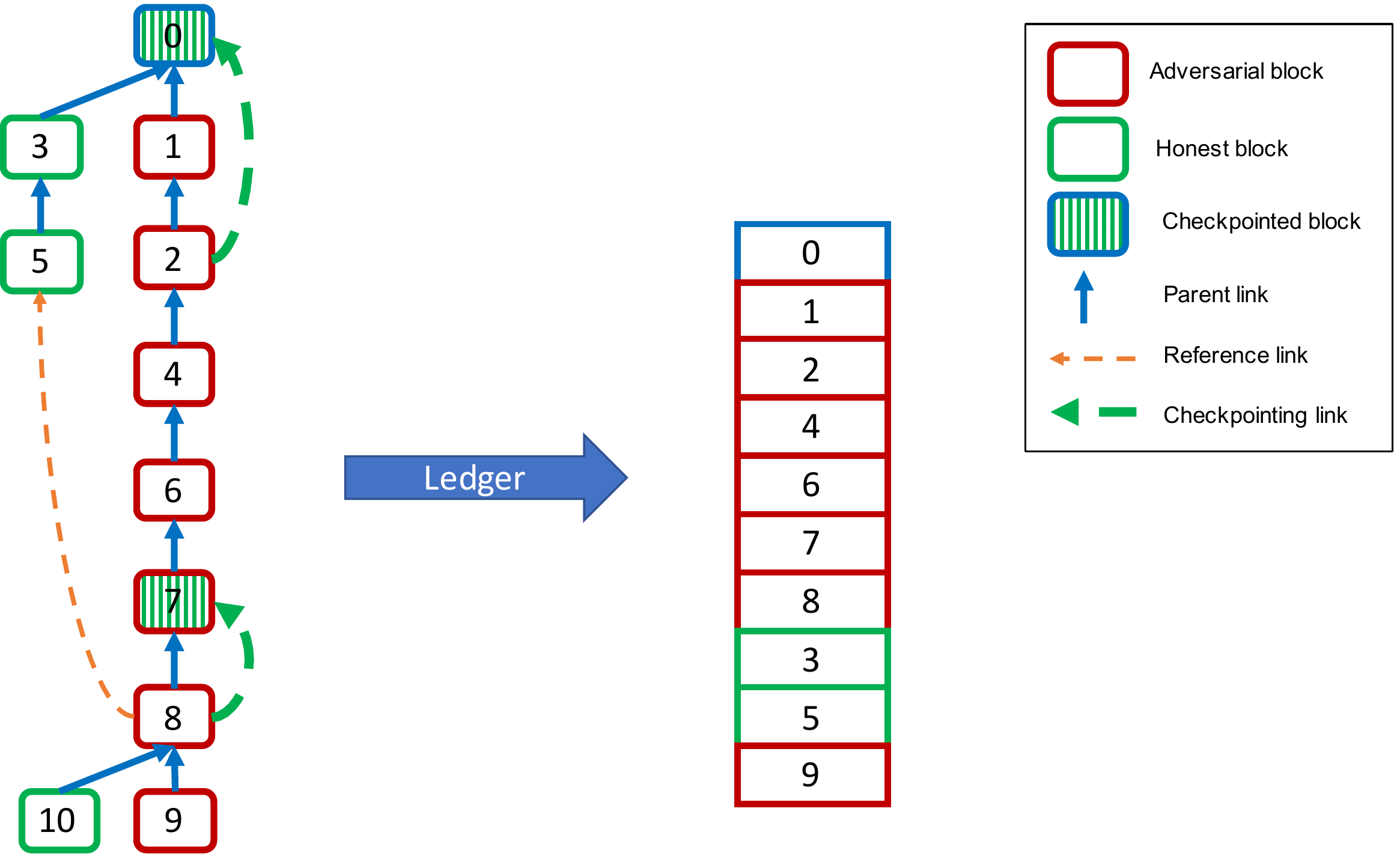}
    \caption{\protocol~ protocol: 30\% honest mining power ensures 30\% of the blocks in the ledger are mined by honest miners.}
    \label{fig:protocol}
\end{figure}

\begin{figure}
    \centering
    \includegraphics[width = .9\columnwidth]{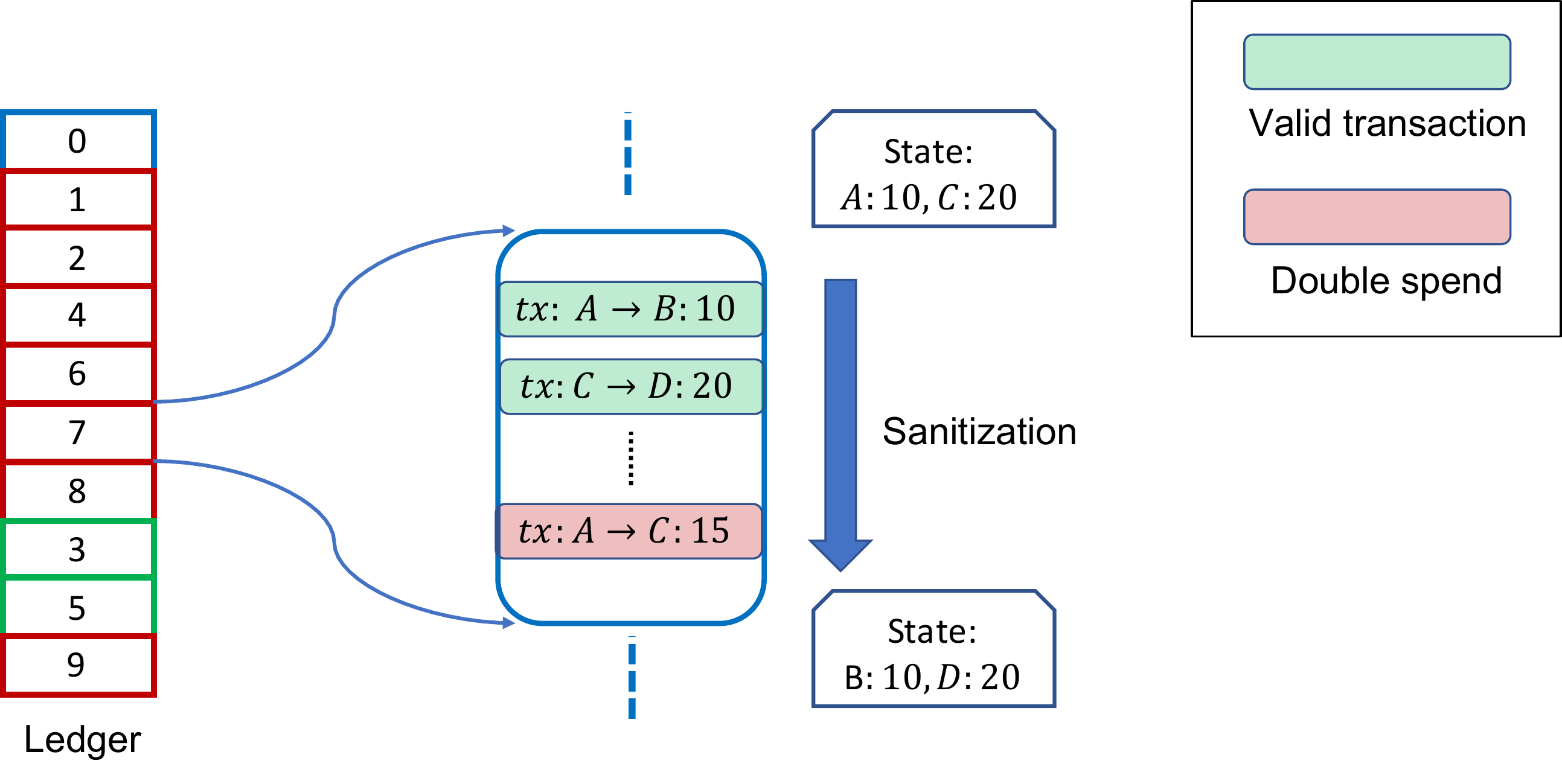}
    \caption{Ledger sanitization ignores/removes invalid transactions post ordering.}
    \label{fig:sanitization}
\end{figure}

\subsection{Security Properties of \protocol}
In this section we show that \protocolnosp, under the centralized setting of a single honest checkpointing party, satisfies a host of desirable properties. First and foremost, Theorems~\ref{thm:advocate-safety} and~\ref{thm:advocate-liveness} prove that \protocol satisfies safety and liveness (cf.\  Section~\ref{subsec:ledger-model}).

\begin{theorem}[Safety]\label{thm:advocate-safety}
    Let $B$ be a block which is checkpointed by \protocol via certificate $C$. If $B$ is part of the main chain, then $B$ is stable (cf.\ Definition~\ref{def:stable}). If $B$ is an uncle, then $B$ is stable if $C$ is published in a block which is $k$-deep, with $k = e-c$.
\end{theorem}

\begin{proof}
    Let $C_i = \{B_i, \mathcal{R}_i, S_i\}$ be the $i$-th certificate and $B$ a block checkpointed by $C_i$. By definition of the protocol, all honest parties eventually accept a chain which contains a referring block $B_C$ which contains $C_i$.
    Observe that, once $C_i$ is created, the ledger position of $B_i$ is finalized, given the ledger construction description in Section~\ref{subsec:decoupled-validity}. Therefore, if $B$ is part of the main chain, its ledger position is also fixed as soon as $C_i$ is created.
    If $B$ is an uncle, then its ledger position depends on the referring block $B_C$. Specifically, the ledger construction rules enforce that the uncle blocks which are checkpointed by $C_i$ are appended in the final ledger \emph{after} the referring block for $C_i$. However, a referring block \emph{can} be reverted, if a chain appears which is both valid (\ie contains a correct referring block) and long enough. Therefore, the position of uncle blocks, which are checkpointed by $C_i$, is finalized only when the certificate $C_{i+1}$ is issued, which occurs after at most $e-c$ main chain blocks.
    \qedsymbol
\end{proof}

\begin{theorem}[Liveness]\label{thm:advocate-liveness}
    Let $h$ be the probability that at least one honest block is created per round; \protocol satisfies liveness (cf. Definition~\ref{def:liveness}) with parameter $\livenessParam = \lceil \frac{2}{h} \rceil \cdot e$.
\end{theorem}

\begin{proof}
    The proof follows directly from the ledger construction (cf. Section~\ref{subsec:ledger-model}) and Theorem~\ref{thm:advocate-safety}. Specifically, let $t$ be the round on when block $B$ is created. 
    If $B$ is a main chain block, then it becomes stable with the issuing of the first checkpoint after $t$ which, by definition of the checkpointing behavior (Section~\ref{subsec:single-node-behavior}), occurs at most $\frac{e}{h}$ rounds after $t$.
    If $B$ is an uncle block then, as shown in Theorem~\ref{thm:advocate-safety}, it becomes stable with the issuing of the first checkpoint \emph{after} $B$ becomes checkpointed; in other words, $B$ becomes stable when $2$ checkpoints are issued after it is created. 
    However, the chain growth depends at worst on the honest miners' mining power (\eg if the adversary abstains), therefore two checkpoints are issued on expectation at most $\lceil \frac{2}{h} \rceil \cdot e$ rounds after $t$. 
    Finally, as mentioned in Section~\ref{subsec:advocate-miner}, a transaction is published in the first honestly-generated block which is produced after the transaction's creation. In turn, this block is checkpointed, either as part of the main chain or as an uncle block, by the upcoming checkpoint. 
    \qedsymbol
\end{proof}

The next property that we explore is chain quality (cf.\  Definition~\ref{def:chain-quality}). First, we observe that \protocol cannot guarantee chain quality over any fixed window of $l$ consecutive blocks of the final ledger. Briefly, the adversary can produce blocks in private and release them, such that the checkpoint certificate refers to all of them at once, hence temporarily flooding the ledger with adversarial blocks. However, as Theorem~\ref{thm:advocate-cq} shows, \protocol \emph{does} guarantee chain quality over the entire ledger. This is a direct improvement of the checkpoint protocol in~\cite{cryptoeprint:2020:173}, which guarantees safety and liveness but not chain quality.

\begin{theorem}[Chain Quality]\label{thm:advocate-cq}
    Let $\beta$ be the adversarial power. For every execution, during which $l$ blocks are created in aggregate by all parties, \protocol satisfies chain quality (cf.\ Definition~\ref{def:chain-quality}) with parameters $l,(1-\beta)$.
\end{theorem}
\begin{proof}
    During the entire execution, the honest parties collectively create (on expectation) at least $(1 - \beta) \cdot l$ blocks. At the end of the execution, a checkpoint is issued, which references all main chain and uncle blocks that are not checkpointed. After the issuing of the last checkpoint, the aggregate ledger at the end of the execution contains all blocks created by all parties, hence the ratio of honest blocks in the final, aggregate ledger is at least $1 - \beta$.
    \qedsymbol
\end{proof}

\paragraph{\protocol with hooks.} To achieve chain quality for smaller windows of blocks, we propose a slightly modified version of \protocolnosp. Now, each block contains a reference to the latest checkpoint certificate $C_j$ at the time it was mined. Next, such block can be referenced by a certificate $C_i$ only if $i - j \leq t$, \ie it can be referenced only by one of the $t$ certificates that immediately follow $C_j$. This constraint, called a \emph{hook}, prevents $\adversary$ from releasing old blocks.

Theorem~\ref{thm:advocate-cq-hooks} shows that \protocol with hooks ensures chain quality for any window of blocks containing $t$ consecutive checkpoints.

\begin{theorem}[Short Term Chain Quality]\label{thm:advocate-cq-hooks}
    Under \protocol \emph{with hooks}, the ratio of honest blocks in any window of $l$ consecutive blocks, which includes $t$ checkpoints, is at least $\frac{(1 - \beta) \cdot (t - 1)}{t + \beta + t \cdot \beta - 1}$, where $\beta$ is the adversarial power.
\end{theorem}

\begin{proof}
    Let $\Upsilon$ be the maximum number of blocks that are produced on expectation by \emph{all} parties (honest and adversarial) between two consecutive checkpoints. 
    Without loss of generality, assume a window of blocks which begins with the checkpointed block of certificate $C_i$ and ends with the checkpointed block of certificate $C_{i+t}$. $C_i$ can reference at most $t \cdot \beta \cdot \Upsilon$ adversarial blocks (\ie which have been created after certificate $C_{i-t}$) and at minimum $0$ honest blocks (\ie if all honest blocks created between certificates $C_{i-1}$ and $C_i$ are part of the main chain). 
    Also, certificate $C_{i+t}$ can reference at most $(1 - \beta) \cdot \Upsilon$ honest blocks and at minimum $0$ adversarial blocks (\ie if all honest blocks created between certificates $C_{i+t-1}$ and $C_{i+t}$ are uncle blocks and all such adversarial blocks are part of the main chain). 
    In this case, the above window of blocks contains $2 \cdot t \cdot \beta \cdot \Upsilon$ adversarial blocks and $(t - 1) \cdot (1 - \beta) \cdot \Upsilon$ honest blocks, hence the ratio of honest blocks is $\frac{(1 - \beta) \cdot (t - 1)}{t + \beta + t \cdot \beta - 1}$.
    \qedsymbol
\end{proof}

Finally, we introduce two performance metrics  that accentuate the functionality of \protocolnosp. First, the \emph{chain inclusion gap} (Definition~\ref{def:chain-inclusion-gap}) expresses the expected number of blocks until a new block is stable. Corollary~\ref{cor:advocate-chain-inclusion} shows that plain \protocol cannot ensure a chain inclusion gap, whereas \protocol with hooks guarantees a chain inclusion gap of $(\beta \cdot t - \beta + 1) \cdot \Upsilon$ blocks, 
where $\Upsilon$ is the maximum number of blocks that all parties produce on expectation between two consecutive checkpoints; the proof follows directly from Theorems~\ref{thm:advocate-cq} and~\ref{thm:advocate-cq-hooks}. We know that the chain grows at the rate $1-\beta$; thus, the maximum number of blocks mined between two checkpoints is $\frac{e}{1-\beta} = \Upsilon$.
Observe that the chain inclusion gap increases linearly with $\Upsilon$ and, consequently, with the epoch length, this is a direct improvement on the result of~\cite{cryptoeprint:2020:173}, where it increases exponentially under adversarial mining majority.

\begin{definition}[Chain Inclusion Gap]\label{def:chain-inclusion-gap}
    Let party $P$ with a main chain $C$ of length $l$, which creates a new block $B$. The chain inclusion gap with parameter $g$ states that, when $B$ becomes stable, its position in the aggregate ledger is at most $l + g$.
\end{definition}

\begin{corollary}[\protocol Chain Inclusion Gap]\label{cor:advocate-chain-inclusion}
    \protocol guarantees chain inclusion gap (cf. Definition~\ref{def:chain-inclusion-gap}) with parameter $g = \infty$. \protocol \emph{with hooks} guarantees chain inclusion gap with parameter $g = (\beta \cdot t - \beta + 1) \cdot \frac{e}{1-\beta}$, where $\beta$ is the adversarial power, $t$ is the hook parameter, and $e$ is the checkpoint epoch length. 
\end{corollary}


Second, \emph{optimistic serializability} (Definition~\ref{def:serializability}) ensures that the checkpointing service does not trivialize the ledger maintenance. Specifically, under fully honest conditions, \ie $\beta=0$, transactions are ordered in the ledger in the order of their arrival, if such arrival order exists. 

\begin{definition}[Optimistic Serializability]\label{def:serializability}
     For two transactions $tx$, $tx'$, where $tx$ was given as an input to all honest nodes at round $r$ and is valid w.r.t. ledger $L_P(r)$ at round $r$ and $tx'$ was given as an input to all honest parties after round $r$, it holds that for any $r'>r$, the ledger $L_P(r')$ of any honest party $P$ cannot include $tx'$,$tx$ in this order, given that the network consists of all honest nodes.
\end{definition}

Permissionless protocols like Nakamoto longest chain ensures optimistic serializability and the checkpointing service does not affect the block ordering; playing a supplementary role. Evidently, \protocol satisfies optimistic serializability by design.

\begin{corollary}[\protocol optimistic serializability]\label{cor:advocate-chain-inclusion}
    \protocol guarantees optimistic serializability (cf. Definition~\ref{def:serializability}).
\end{corollary}

Consider a Nakamoto longest chain protocol; it is easy to show that it guarantees optimistic serializability. Since the honest nodes received $tx$ before $tx'$, all miners will mine ledger with $tx$ before $tx'$. Even when the ledger is forked, within each fork, the parent is known and hence the order is maintained for all parties. In \protocolnosp, it may happen that the certificate refers to a transaction $tx$ again; hence the ledger $L_P(r)$ might have transactions $tx$, $tx'$ in that order. However, since the base consensus is Nakamoto; it will ensure that $tx$, $tx'$ exists in that order before checkpointing. Thus the $tx'$ referred by the checkpoint $C_r$ will be a second occurrence and will be removed by ledger sanitization.

\subsection{Contrast with Fruitchains and Conflux}
In terms of safety, the transaction inclusion from $F_i$ has similarity to fruits in Fruitchains \cite{fruitchains} and DAG references in the pivot chain of Conflux \cite{conflux}. However,  we note that if the adversarial mining power is greater than 50\%, the adversary can always beat honest nodes by creating a mainchain in Fruitchain and a conflicting pivot chain in Conflux, thus violating safety.

In terms of liveness, Fruitchains has a chain quality of $1-\beta$, which is ensured because all blocks created by honest miners are eventually included as fruits. This chain quality however is reduced to $0$ for $\beta>0.5$. To understand this abrupt loss of chain quality, let us consider an attack by a $51\%$ adversary. The adversary creates a longer blockchain, consisting of only adversarial blocks, and the adversarial blocks do not include references to fruits mined by honest miners. Since the longest blockchain is chosen to create the fruit ledger, it will not consist of any honest fruits, rendering the chain quality $0$. Similar arguments can be made for a heavier pivot chain generated by adversaries in Conflux. This abrupt loss of chain quality for $\beta>0.5$ is depicted in Figure \ref{fig:cq_comparison}.

Note that implementing existing checkpointing designs like Casper to Fruitchains will not will not improve chain quality since the adversarial majority can ensure that the main chain only contains blocks referring to no honest fruits; hence even the checkpointed ledger will not contain any honest fruits. Instead, Advocate checkpoints all leaves of the block tree; this achieves optimal chain quality (Theorem 3) because the final ledger consists of all produced blocks, where the percentage of adversarial blocks is equal to the adversarial mining power (the lower bound).
\section{\protocol with BFT checkpointing}
\label{sec:bft}

Although \protocolnosp, as described above, satisfies the desired security properties, it assumes a single checkpointing node. This centralized design is problematic, especially in systems like distributed ledgers, whose main purpose is decentralization. In this section, we present \protocol with Byzantine Fault Tolerant (BFT) checkpointing, which extends the single checkpointing node with a committee of $n$ nodes. Although this extension might seem trivial, there are certain fine points (for example running a BFT-SMR with external state validation) that need to be analyzed to establish equivalency. In contrast to other checkpointing protocols, \protocol-BFT allows multiple checkpointing candidates, since a candidate includes both the checkpointed block and the reference links, thus increasing the input space for the BFT committee nodes. The committee achieves consensus on the contents of the checkpoint certificate, which is then published on the main chain.

We present a design that optimizes transaction inclusion and confirmation latency. 
The adversary $\adversary$ controls up to $f$ committee nodes, s.t. $n \geq 3f+1$, and a fraction $\beta \in [0, 1)$ of the PoW mining power. Note that the BFT committee is independent from the committee of miners. Hence, it is possible for the committee of miners (\ie a small community for novel PoW chains) to be in adversarial majority, whereas the BFT committee (consisting of well established and legally bound validators) is in honest supermajority. 

\subsection{\protocolnosp-BFT}


The committee nodes act as full nodes for the PoW main chain and run a separate SMR (BFT-based State Machine Replication) protocol; 
any generic BFT-SMR protocol should suffice. On receiving a valid PoW block, the committee node posts a transaction $\langle Blockhash, Depth \rangle$ on the SMR chain, which is finalized after some rounds as per the BFT-SMR's rules.

The SMR chain announces a new checkpoint when a transaction containing a block with depth $e$ more than the previous checkpoint is posted on the SMR chain. A checkpoint transaction $tC_i = \{H(B_i),M(R_i)\}$ is posted on the SMR chain, where $R_i$ consists of all the main chain blocks referenced on the SMR chain between the references for $B_{i-1}$ and $B_i$, $M(R_i)$ denotes it's Merkle root and $H(B_i)$ denotes hash of block $B_i$.

The checkpoint certificate $C_i = \{\{tC_i\},R_i,w_i\}$ consists of the checkpoint transaction $tC_i = {B_i,M(R_i)}$, a witness $w_i$ stating that it is finalized on the SMR chain, and the list of references $R_i$. $C_i$ should be posted on the PoW chain before the depth of $d(B_i)+c$, where $d(B_i)$ is the depth of checkpoint $B_i$. 

The SMR chain's block inclusion validity rules are as follows:
\begin{itemize}
    \item Data availability: The block should be available.
    \item Block validity: The block should have valid PoW; 
    note that full transaction validity is not required due to main chain's ledger sanitization. 
    \item No checkpoint conflict: The block should not be at a height of $d(B_i)+e$ and extend a chain that does not contain $B_i$.
\end{itemize}

Evidently, the BFT checkpointing service realizes in a distributed manner the checkpointing node of Section \ref{sec:ideal}. Specifically, the BFT service collects all leaves which are not checkpointed, including the main chain and uncle blocks, and issues a certificate which references them. Assuming the BFT protocol is secure (safe and live), the committee will i) issue a certificate which ii) references all non-checkpointed blocks, so the analysis of Section \ref{sec:ideal} also applies here. 

\paragraph{Main chain behavior}
The main chain miners act as light nodes for the SMR chain. They include $C_i$ in the main chain as soon as $tC_i$ is finalized on the SMR chain. We assume that the main chain nodes are connected to at least one honest SMR chain node, to get the references from $M(R_i)$.
The mining behavior and validity rules including the $c$ constraint on checkpoint inclusion remain the same as described in Section \ref{sec:ideal}.

\paragraph{Latency}
To compute the latency for a transaction, let $\tau_m$ be the time until a block $B$ containing the transaction is mined. Also let $\tau_t$ be the time until a transaction containing the hash of $B$ is posted on the SMR chain and $\tau_f$ be the time until that transaction is finalized on the SMR chain. The total time until an honest transaction is considered for checkpointing is  $\tau_i = \tau_m + \tau_t + \tau_f$. Observe that the value $\tau_t$ is not affected by the adversarial mining fraction.
Finally, the transaction is confirmed when the next checkpoint is posted on the SMR chain, \ie after time $\tau_c$ until the checkpoint is finalized. Therefore, the overall latency of a transaction is $\tau = \tau_i + \tau_c$. 
We note that, the parameter $c$ depends on the BFT latency, \ie $\Delta_{BFT}$. 

\subsection{BFT integration}
We  abstract the BFT functionality $F_{BFT}$ and show equivalence with the checkpointing service $F_{cps}$ described in Section \ref{sec:ideal}.

\vspace{0.3cm}
\begin{mdframed}[backgroundcolor=black!10,rightline=false,leftline=false]

\begin{center}
    {\bf BFT-SMR service $F_{BFT}$}
\end{center}

    \vspace{0.1cm}
    
$F_{BFT}$ is a part of a network of $P$ replicas participating in BFT-SMR. $F_{BFT}$ takes an input $I_{BFT}$ and outputs $O_{BFT}$ after a delay bounded by $\Delta_{BFT}$. $F_{BFT}$ may take no input and still output $O_{BFT}$ depending on the state $S_{BFT}$ and $I_{lBFT}$: an input received by some replica. $F_{BFT}$ checks validity of $I_{BFT}$ and/or $I_{lBFT}$ with respect to $S_{BFT}$ stored locally using $V_{BFT}$.
A message from one replica implementing $F_{BFT}$ to another takes a maximum delay of $\Delta$.
  
\end{mdframed}
\vspace{0.3cm}

We now establish an equivalence between $F_{BFT}$ and $F_{cps}$. {\em inputValue} in $F_{cps}$ is equivalent to $I_{BFT}$ in $F_{BFT}$, while {\em commitDecision} is equivalent to $O_{BFT}$. However, a major difference is that the delay between the two is $\Delta_{BFT}$ in $F_{BFT}$ and $\Delta_{cps}$ in $F_{cps}$, s.t. $\Delta_{BFT} / \Delta_{cps} > 1$. Note that $\Delta_{BFT}$ is dependent on the liveness parameter $u$ of the BFT protocol. A further major deviation regards to the implementation of $V_{BFT}$, which corresponds to the {\em isInputValid} and {\em inputValidity} messages in $F_{cps}$. Since the state is not in the same module and the input may be indirectly received from a different replica, the data needed for validating $I_{lBFT}$ may not be available to $F_{advocate}$, thus returning {\em inputValidity} may take unknown time. This is resolved via a unified network functionality $\mathcal{N}_{uni}$ (Figure \ref{fig:network_abstractor}) to which each BFT replica connects.

\vspace{0.1cm}
\begin{mdframed}[backgroundcolor=black!10,rightline=false,leftline=false]

\begin{center}
    {\bf Unified network functionality $\mathcal{N}_{uni}$}
\end{center}

    \vspace{0.1cm}
    
    $\mathcal{N}_{uni}$ gets messages from the nodes connected to it. The message handling is split in 3 levels. The first level is the {\em Network Handler}, which manages network functions like message downloading and forwarding. Once the message is downloaded, it is passed to the {\em Validity Handler}, which verifies the message \wrt a well-defined validity predicate $V_{BFT}$, which utilizes {\em isInputValid} and {\em inputValidity}. Once the checks pass, the validity handler forwards the message to $F_{BFT}$, marking the message as {\em received}.
    
\end{mdframed}
\vspace{0.3cm}

Note that $\mathcal{N}_{uni}$ ensures that, once a message is received by $F_{BFT}$, checking $V_{BFT}$ is instantaneous, thus replicating a local state. Moreover, $\mathcal{N}_{uni}$ does not change the synchronous setting delay. Specifically, let $i$ be the first honest node to receive a message $m$ at time $t$. By definition, $F_{\protocol}$ sends {\em inputValidity} to node $i$ at some time $t' \leq t$. Now, the message $m$ propagates across all nodes in a synchronous manner, hence it is downloaded by each honest node $j$ at the latest at time $t+\Delta$. Therefore, $F_{advocate}$ is queried by node $j$ regarding {\em isInputValid} of $m$ at time $t+\Delta$. Since $F_{\protocol}$ replied {\em inputValidity} to node $i$ at time $t$, it will also reply {\em inputValidity} to $j$ at $t+\Delta$. Thus, $m$ is marked as \emph{received} by node $j$ by time $t+\Delta$, ensuring that the network is $\Delta$-synchronous under $\mathcal{N}_{uni}$.

\begin{figure}
    \centering
    \includegraphics[width = .4\columnwidth]{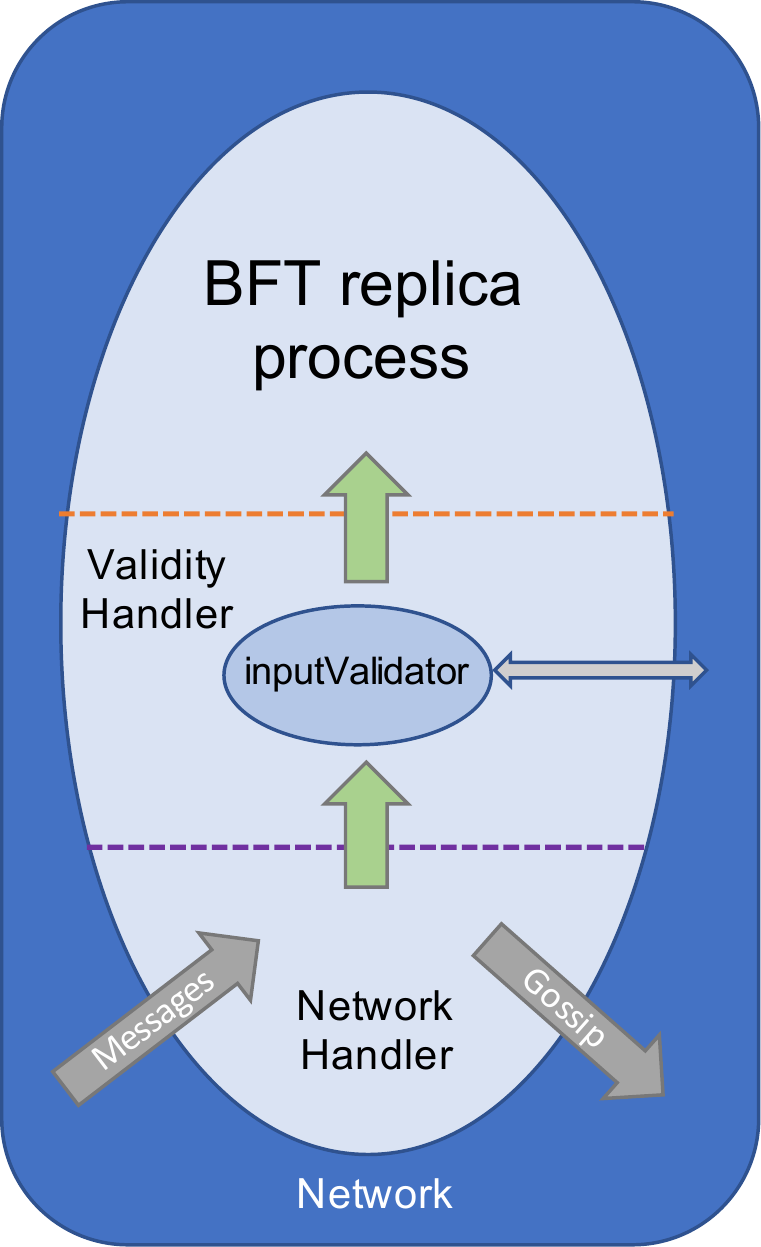}
    \caption{The network layer of BFT is modified to accommodate the external state validation (isInputValid) functionality; the end impact on BFT replica process is equivalent to being connected to a synchronous network without $F_{Advocate}$}
    \label{fig:network_abstractor}
    \vspace{-4mm}
\end{figure} 

\paragraph{Properties of the BFT service}

As discussed above, as long as the BFT-SMR protocol is secure, \ie satisfies safety and liveness, it securely realizes the single checkpointing node in a distributed manner. However, our integration of the federation BFT into checkpoints affects the values of the parameters $c, e$. Specifically, $c$ should be at least $\tau_f / \tau_r$ larger than in Section~\ref{sec:ideal}, $\tau_f$ being the time required for a transaction to be finalized by the BFT and $\tau_r$ being the size of each round of the ledger protocol. Therefore, faster BFT protocols are preferable, in order to minimize the time until the checkpoint certificate is finalized. 

In the next section, we present a variant of Advocate for multi-chain systems. 
\section{\protocolnosp: Checkpointing for Parallel-PoW chains}
\label{sec:prism}

Many emerging PoW blockchains rely on a ``parallel-chain'' architecture for scaling, where multiple chains run in parallel and are aggregated. Two successful parallel-chain architectures are  {\sf Prism} \cite{BagariaKTFV19} and  {\sf OHIE}  \cite{yu2019ohie}. Although these two protocols are significantly different from each other, we demonstrate the generalizability of \protocolnosp; we extend \protocol to both these settings by proposing a {\em meta-protocol}  \protocolnosp-PC for integrating \protocol to parallel-chain architectures. 
For simplicity, we design \protocolnosp-PC using a single (honest) checkpointing node, which can be readily extended to a BFT federation as described in Section \ref{sec:bft}.

\subsection{\protocolnosp-PC: Meta-protocol}
Consider $M$ parallel chains, with mining power {\em sortition} across them. A block is labelled as $B_{m,j,f_b}$ if it belongs to branch $f_b$ of chain $m$ and has rank $j$. $f_b$ is a function of parent of $B_{m,j,f_b}$. The rank $\calR(B_{m,j,f_b})=j$ of a block is determined by the parallel-chain protocol's specifics and is deterministic when a block is mined. We highlight the first important meta-principle: 
\begin{quote}
 {\bf Rank criterion:} {\em Blocks mined by honest miners have monotonically increasing rank within a chain}.    
\end{quote}
Block ranks are used to determine epoch intervals with a checkpointing epoch spanning blocks with a rank-difference of $e$. 
We denote a chain as {\em payload-carrying}, denoted by $Y(i)=1$ if its blocks are designed to carry a transaction payload. We set one chain (chain $0$) as a {\em base-chain}, which may or may not be payload-carrying.\\ 

\noindent{\bf Checkpointing party behavior}. 
\begin{itemize}
    \item Upon receiving Block $B_{0,j,f_b}$ at rank $\calR(\tilde{C}_{i-1})+e$, where $\tilde{C}_{i-1}$ is the latest checkpointed base-chain block, it creates a new checkpoint certificate $C_i$, the block $B_{0,j,f_b}$ is the checkpointed block $\tilde{C_i}$; 
    \item The certificate  $C_i=\{R(\tilde{C_i}),R_i,\mathbf{B_i}\}$ defines: 
    (a) a vector of $M-1$ parallel-chain blocks $\mathbf{B_i}$, \ie one tip block from each parallel chain except the base-chain,
    (b) the reference list $R_i$ of all payload-carrying blocks not referenced by any checkpoint until Rank $\calR(\tilde{C_i})$, 
    (c) a reference $R(\tilde{C}_i)$ to the checkpointed base-chain block $\tilde{C}_i$. 
\end{itemize}

\noindent{\bf Validity rules}. 
\begin{itemize}
    \item All chains extend the latest checkpoint;
    \item A base-chain block is invalid if it extends the chain past rank $\calR(\tilde{C}_i)+c$ and does not contain $C_i$; 
    \item A non base-chain is invalid if none of it's blocks refer to the base-chain block containing the certificate $C_i$ by rank $\calR(\tilde{C}_i)+c$; 
    \item A non base-chain tip block $B_{m,j,f_b}$ is valid for inclusion in $\mathbf{B_i}$ only if chain $m$ has referred to $C_{i-1}$.
\end{itemize}

The ledger creation rules are similar to \protocolnosp; the checkpoint certificate brings in all the referred blocks $R_i$ in the respective payload-carrying chain's ledger. 

{\em \protocolnosp-OHIE:} We observe that \texttt{rank} in {\sf OHIE} satisfies the rank criterion and thus can be treated as $\calR$. The meta-protocol can be integrated into OHIE by setting the chain 0 as base chain and assigning its protocol defined \texttt{rank} with the \protocol rank $\mathcal{R}$. A similar extension of the rank criterion works for ledger combiner as well.

{\em \protocolnosp-Prism:} We observe that Prism's proposer levels follow the {\em Rank criterion} since honest miners always mine blocks with increasing proposer levels. We can integrate \protocolnosp-PC by setting the Proposer chain as the base-chain and assigning proposer level as \protocol rank $\mathcal{R}$.

\section{Implementation and Evaluation}
\label{sec:implementation}
\label{sec:evaluation}

We implement \protocol on a codebase in \texttt{Rust} and compare its performance with various existing checkpointing techniques. To test the performance of \protocol to the limit, we integrate \protocolnosp-Prism along with the high performance implementation of {\sf Prism} written in \texttt{Rust} \cite{prism_systems}. We evaluate the performance of \protocol comparing with various other checkpointing protocols. The code is available at \cite{advocate_sys_anon}.

\subsection{Comparison baselines}

We  implement two other checkpointing protocols as baselines to  compare performance metrics of \protocol. We briefly describe these baselines and their integration with {\sf Prism} below. 

\noindent{\bf \randomnessnosp-checkpointing}. Derived from  \cite{cryptoeprint:2020:173}, the checkpoint certificates referring to a single Block-hash(checkpoint) are introduced in the ledger. The certificates add randomness at every epoch, ensuring the adversary cannot implement a front-running attack described in \cite{cryptoeprint:2020:173}. \randomnessnosp-checkpointing is implemented by modifying the \protocol $F_{cps}$ to generate checkpoint without references. 

\noindent{\bf \basicnosp-checkpointing}. Derived from the off-chain checkpoints published by Nakamoto in the early days of bitcoin (checkpointing via GitHub). The checkpoints certificates are posted off-chain and consist of a block's hash. The full node codebase recognizes these checkpoints and only considers chains extending these checkpoints as valid. \basicnosp-checkpointing is implemented by modifying code for \randomnessnosp-checkpointing not to include the certificate on chain.

\noindent{\bf Experimental setup:} We run \protocol and \protocolnosp-Prism experiments on c5d.large and c5d.4xlarge AWS instances respectively. We run our experiments with $\beta\geq 0.5$ with private mining attack where the adversary mines a private chain and broadcast private blocks in bursts of epoch length. Our evaluation answers the following questions:
\begin{enumerate}
    \item How do performance and security metrics of \protocol compare to state-of-the-art  checkpointing and finality gadgets?
    \item How does the performance of \protocol react to  a slow checkpointing service?
    \item How does \protocol perform with large epoch sizes? 
    \item How does \protocol integrate with  very high throughput  PoW blockchains, e.g.,   {\sf Prism}? How is the performance overhead?
\end{enumerate}

\begin{table*}
\vspace{-4mm}
\centering
\begin{tabular}{|*{12}{c|}}
\hline
\multicolumn{3}{|c}{\bf Parameters} & \multicolumn{3}{|c}{$e=5$, $\Delta_{BFT}$=0} & \multicolumn{3}{|c}{$e=5$, $\Delta_{BFT}$=2} & \multicolumn{3}{|c|}{$e=10$, $\Delta_{BFT}$=0} \\ \hline
\multicolumn{3}{|c|}{\bf Metrics} & $FG$ & $IL$ & $HW$ & $FG$ & $IL$ & $HW$ & $FG$ & $IL$ & $HW$ \\ \hline
\multicolumn{3}{|c|}{\basicnosp-cp} & 0.148 & 67.14 & 0.7032 & - \footnote[2]{Experiments discontinued due to established shortcomings of the protocol} & - & - & - & - & - \\\hline
\multicolumn{3}{|c|}{\randomnessnosp-cp} & 0.323 & 6.688 & 0.3539 & 0.204 & 13.76 & 0.594 & 0.227 & 21.53 & 0.546 \\\hline
\multicolumn{3}{|c|}{\protocol} & \bf 0.588 & \bf 3.611 & \bf0 & \bf0.514 &  \bf 2.546 & \bf 0 & \bf 0.475 & \bf 6.712 & \bf 0.048 \\\hline
\end{tabular}
\caption{\protocol evaluation for $\beta=0.5$}
\label{tab:50}
\vspace{-7mm}
\end{table*}

\begin{table*}
\centering
\begin{tabular}{|*{12}{c|}}
\hline
\multicolumn{3}{|c}{\bf Parameters} & \multicolumn{3}{|c}{$e=5$, $\Delta_{BFT}$=0} & \multicolumn{3}{|c}{$e=5$, $\Delta_{BFT}$=2} & \multicolumn{3}{|c|}{$e=10$, $\Delta_{BFT}$=0} \\ \hline
\multicolumn{3}{|c|}{\bf Metrics} & $FG$ & $IL$ & $HW$ & $FG$ & $IL$ & $HW$ & $FG$ & $IL$ & $HW$ \\ \hline
\multicolumn{3}{|c|}{\basicnosp-cp} & 0 & $\infty$ & 1 & - \footnote[2]{Experiments discontinued due to established shortcomings of the protocol} & - & - & - & - & -  \\\hline
\multicolumn{3}{|c|}{\randomnessnosp=cp} & 0.101 & 43.11 & 0.696 & 0.033 & 24.43 & 0.9 & 0 & $\infty$ & 1 \\\hline
\multicolumn{3}{|c|}{\protocol} & \bf 0.389 & \bf 3.491 & \bf 0 & \bf 0.330 & \bf 3.217 & \bf 0 & \bf 0.311 & \bf 6.512 & \bf 0.056 \\\hline
\end{tabular}
\caption{\protocol evaluation for $\beta=0.67$}
\label{tab:67}
\vspace{-8mm}
\end{table*}

\begin{table*}
\centering
\begin{tabular}{|*{12}{c|}}
\hline
\multicolumn{3}{|c}{\bf Parameters} & \multicolumn{3}{|c}{$e=5$, $\Delta_{BFT}$=0} & \multicolumn{3}{|c}{$e=5$, $\Delta_{BFT}$=2} & \multicolumn{3}{|c|}{$e=10$, $\Delta_{BFT}$=0} \\ \hline
\multicolumn{3}{|c|}{\bf Metrics} & $FG$ & $IL$ & $HW$ & $FG$ & $IL$ & $HW$ & $FG$ & $IL$ & $HW$ \\ \hline
\multicolumn{3}{|c|}{\basicnosp-cp} & 0 & $\infty$ & 1 & - \footnote[2]{Experiments discontinued due to established shortcomings of the protocol} & - & - & - & - & -  \\\hline
\multicolumn{3}{|c|}{\randomnessnosp-cp} & 0 & $\infty$ & 1 & 0 & $\infty$ & 1 & 0 & $\infty$ & 1 \\\hline
\multicolumn{3}{|c|}{\protocol} & \bf 0.102 & \bf 2.849 & \bf 0 & \bf 0.072 & \bf 4.319 & \bf 0.278 & \bf 0.087 & \bf 6.467 & \bf 0.13 \\\hline
\end{tabular}
\caption{\protocol evaluation for $\beta=0.9$}
\label{tab:90}
\vspace{-7mm}
\end{table*}

\noindent{\bf Performance metrics:} We use three metrics defined below to measure performance:
\begin{itemize}
    \item {\it Fractional Goodput ($FG$):} Let Goodput ($\mathcal{G}$) be the number of honest transactions confirmed per unit time and optimal throughput ($\mathcal{T}$) be the maximum throughput in the absence of an adversary. We define fractional goodput as  $\mathcal{G}/\mathcal{T}$.
    \item {\it Ledger Inclusion latency ($IL$):} for an honest party $P$ is the time taken (measure in means of block arrival time $\Delta_A$) between transaction generation and inclusion in the ledger of $P$.
    \item {\it Honest block wastage ($HW$):} Fraction of honest blocks that are not part of the ledger.
\end{itemize}
The performance metrics are tabulated in Tables \ref{tab:50}, \ref{tab:67}, \ref{tab:90} for a variety of  experimental settings: varying  adversary mining power, BFT network latency ($\Delta_{BFT}$), epoch size.  Each experiment was conducted over a range of  50-100 epochs. We make the following more broad observations from the data.
\begin{enumerate}
    \item We observe that \protocol is far superior to it's competitors in all settings and metrics
    \item Advocate takes a lesser hit on performance as compared to its competitors if the checkpointing service is slow
    \item $IL$ of Advocate increases linearly with epoch length with minimal drop in its $FG$
\end{enumerate}

\paragraph{Prism} The Prism full-stack implementation can achieve a throughput of 70K tx/s; coupled with \protocol an optimized implementation should achieve a throughput of $(1-\beta) \cdot 70K$ tx/s. However, checkpointing adds to the validity rules, leading to loss of throughput. Our implementation of \protocolnosp-Prism aims to develop a proof-of-concept, s.t. checkpointing can be integrated into high throughput parallel chains without much expense of throughput. To this end, we design an adversary $\adversary$ for Prism who censors honest transaction blocks. For such $\adversary$ with $\beta=0.7$, our full-stack implementation of \protocolnosp-Prism (with UTXO at the application layer and p2p networking at the networking layer) can achieve a goodput of {\em 8200 tx/s} with a confirmation latency of 120s. While 8200 tx/s with a 70\% adversary is much higher than any existing protocol can achieve, we believe it can be further improved by optimizing interaction between ledger manager and integrated checkpointing module, a direction which will be explored in future research. 


\section{Acknowledgements}
This research was partly supported by US Army Research Office Grant W911NF-18-1-0332, National Science Foundation CCF- 1705007, NeTS 1718270, XDC network and IOHK.

\bibliographystyle{ACM-Reference-Format}
\bibliography{references}


\begin{thebibliography}{34}


\ifx \showCODEN    \undefined \def \showCODEN     #1{\unskip}     \fi
\ifx \showDOI      \undefined \def \showDOI       #1{#1}\fi
\ifx \showISBNx    \undefined \def \showISBNx     #1{\unskip}     \fi
\ifx \showISBNxiii \undefined \def \showISBNxiii  #1{\unskip}     \fi
\ifx \showISSN     \undefined \def \showISSN      #1{\unskip}     \fi
\ifx \showLCCN     \undefined \def \showLCCN      #1{\unskip}     \fi
\ifx \shownote     \undefined \def \shownote      #1{#1}          \fi
\ifx \showarticletitle \undefined \def \showarticletitle #1{#1}   \fi
\ifx \showURL      \undefined \def \showURL       {\relax}        \fi
\providecommand\bibfield[2]{#2}
\providecommand\bibinfo[2]{#2}
\providecommand\natexlab[1]{#1}
\providecommand\showeprint[2][]{arXiv:#2}

\bibitem[anon(2021)]%
        {advocate_sys_anon}
\bibfield{author}{\bibinfo{person}{anon}.} \bibinfo{year}{2021}\natexlab{}.
\newblock \bibinfo{title}{Advocate System implementation}.
\newblock
  \bibinfo{howpublished}{\url{https://github.com/advocate-checkpoint/advocate-systems}}.
\newblock


\bibitem[Azouvi et~al\mbox{.}(2019)]%
        {winkle}
\bibfield{author}{\bibinfo{person}{Sarah Azouvi}, \bibinfo{person}{George
  Danezis}, {and} \bibinfo{person}{Valeria Nikolaenko}.}
  \bibinfo{year}{2019}\natexlab{}.
\newblock \bibinfo{title}{Winkle: Foiling Long-Range Attacks in Proof-of-Stake
  Systems}.
\newblock \bibinfo{howpublished}{Cryptology ePrint Archive, Report 2019/1440}.
\newblock
\newblock
\shownote{\url{https://eprint.iacr.org/2019/1440}}.


\bibitem[Azouvi et~al\mbox{.}(2018)]%
        {azouvi18}
\bibfield{author}{\bibinfo{person}{Sarah Azouvi}, \bibinfo{person}{Patrick
  McCorry}, {and} \bibinfo{person}{Sarah Meiklejohn}.}
  \bibinfo{year}{2018}\natexlab{}.
\newblock \showarticletitle{Betting on Blockchain Consensus with Fantomette}.
\newblock \bibinfo{journal}{\emph{CoRR}}  \bibinfo{volume}{abs/1805.06786}
  (\bibinfo{year}{2018}).
\newblock
\showeprint[arxiv]{1805.06786}
\urldef\tempurl%
\url{http://arxiv.org/abs/1805.06786}
\showURL{%
\tempurl}


\bibitem[Bagaria et~al\mbox{.}(2019)]%
        {BagariaKTFV19}
\bibfield{author}{\bibinfo{person}{Vivek~Kumar Bagaria},
  \bibinfo{person}{Sreeram Kannan}, \bibinfo{person}{David Tse},
  \bibinfo{person}{Giulia~C. Fanti}, {and} \bibinfo{person}{Pramod Viswanath}.}
  \bibinfo{year}{2019}\natexlab{}.
\newblock \showarticletitle{Prism: Deconstructing the Blockchain to Approach
  Physical Limits}. In \bibinfo{booktitle}{\emph{Proceedings of the 2019 {ACM}
  {SIGSAC} Conference on Computer and Communications Security, {CCS} 2019,
  London, UK, November 11-15, 2019}},
  \bibfield{editor}{\bibinfo{person}{Lorenzo Cavallaro},
  \bibinfo{person}{Johannes Kinder}, \bibinfo{person}{XiaoFeng Wang}, {and}
  \bibinfo{person}{Jonathan Katz}} (Eds.). \bibinfo{publisher}{{ACM}},
  \bibinfo{pages}{585--602}.
\newblock
\urldef\tempurl%
\url{https://doi.org/10.1145/3319535.3363213}
\showDOI{\tempurl}


\bibitem[Buterin and Griffith(2017)]%
        {buterin2017casper}
\bibfield{author}{\bibinfo{person}{Vitalik Buterin} {and}
  \bibinfo{person}{Virgil Griffith}.} \bibinfo{year}{2017}\natexlab{}.
\newblock \bibinfo{title}{Casper the friendly finality gadget}.
\newblock
\newblock


\bibitem[Buterin et~al\mbox{.}(2019)]%
        {casper-incentives}
\bibfield{author}{\bibinfo{person}{Vitalik Buterin},
  \bibinfo{person}{Dani{\"{e}}l Reijsbergen}, \bibinfo{person}{Stefanos
  Leonardos}, {and} \bibinfo{person}{Georgios Piliouras}.}
  \bibinfo{year}{2019}\natexlab{}.
\newblock \showarticletitle{Incentives in Ethereum's Hybrid Casper Protocol}.
  In \bibinfo{booktitle}{\emph{{IEEE} International Conference on Blockchain
  and Cryptocurrency, {ICBC} 2019, Seoul, Korea (South), May 14-17, 2019}}.
  \bibinfo{publisher}{{IEEE}}, \bibinfo{pages}{236--244}.
\newblock
\urldef\tempurl%
\url{https://doi.org/10.1109/BLOC.2019.8751241}
\showDOI{\tempurl}


\bibitem[Castro and Liskov(1999)]%
        {castro1999practical}
\bibfield{author}{\bibinfo{person}{Miguel Castro} {and}
  \bibinfo{person}{Barbara Liskov}.} \bibinfo{year}{1999}\natexlab{}.
\newblock \showarticletitle{Practical Byzantine Fault Tolerance}. In
  \bibinfo{booktitle}{\emph{Proceedings of the Third {USENIX} Symposium on
  Operating Systems Design and Implementation (OSDI), New Orleans, Louisiana,
  USA, February 22-25, 1999}}, \bibfield{editor}{\bibinfo{person}{Margo~I.
  Seltzer} {and} \bibinfo{person}{Paul~J. Leach}} (Eds.).
  \bibinfo{publisher}{{USENIX} Association}, \bibinfo{pages}{173--186}.
\newblock
\urldef\tempurl%
\url{https://dl.acm.org/citation.cfm?id=296824}
\showURL{%
\tempurl}


\bibitem[Civit et~al\mbox{.}(2019)]%
        {polygraph}
\bibfield{author}{\bibinfo{person}{Pierre Civit}, \bibinfo{person}{Seth
  Gilbert}, {and} \bibinfo{person}{Vincent Gramoli}.}
  \bibinfo{year}{2019}\natexlab{}.
\newblock \bibinfo{title}{Polygraph: Accountable Byzantine Agreement}.
\newblock \bibinfo{howpublished}{Cryptology ePrint Archive, Report 2019/587}.
\newblock
\newblock
\shownote{\url{https://eprint.iacr.org/2019/587}}.


\bibitem[Daian et~al\mbox{.}(2019)]%
        {FC:DaiPasShi19}
\bibfield{author}{\bibinfo{person}{Phil Daian}, \bibinfo{person}{Rafael Pass},
  {and} \bibinfo{person}{Elaine Shi}.} \bibinfo{year}{2019}\natexlab{}.
\newblock \showarticletitle{Snow White: Robustly Reconfigurable Consensus and
  Applications to Provably Secure Proof of Stake}. In
  \bibinfo{booktitle}{\emph{Financial Cryptography and Data Security - 23rd
  International Conference, {FC} 2019, Frigate Bay, St. Kitts and Nevis,
  February 18-22, 2019, Revised Selected Papers}}
  \emph{(\bibinfo{series}{Lecture Notes in Computer Science},
  Vol.~\bibinfo{volume}{11598})}, \bibfield{editor}{\bibinfo{person}{Ian
  Goldberg} {and} \bibinfo{person}{Tyler Moore}} (Eds.).
  \bibinfo{publisher}{Springer}, \bibinfo{pages}{23--41}.
\newblock
\urldef\tempurl%
\url{https://doi.org/10.1007/978-3-030-32101-7\_2}
\showDOI{\tempurl}


\bibitem[David et~al\mbox{.}(2018)]%
        {EC:DGKR18}
\bibfield{author}{\bibinfo{person}{Bernardo David}, \bibinfo{person}{Peter
  Gazi}, \bibinfo{person}{Aggelos Kiayias}, {and} \bibinfo{person}{Alexander
  Russell}.} \bibinfo{year}{2018}\natexlab{}.
\newblock \showarticletitle{Ouroboros Praos: An Adaptively-Secure,
  Semi-synchronous Proof-of-Stake Blockchain}. In
  \bibinfo{booktitle}{\emph{Advances in Cryptology - {EUROCRYPT} 2018 - 37th
  Annual International Conference on the Theory and Applications of
  Cryptographic Techniques, Tel Aviv, Israel, April 29 - May 3, 2018
  Proceedings, Part {II}}} \emph{(\bibinfo{series}{Lecture Notes in Computer
  Science}, Vol.~\bibinfo{volume}{10821})},
  \bibfield{editor}{\bibinfo{person}{Jesper~Buus Nielsen} {and}
  \bibinfo{person}{Vincent Rijmen}} (Eds.). \bibinfo{publisher}{Springer},
  \bibinfo{pages}{66--98}.
\newblock
\urldef\tempurl%
\url{https://doi.org/10.1007/978-3-319-78375-8\_3}
\showDOI{\tempurl}


\bibitem[Dinsdale{-}Young et~al\mbox{.}(2020)]%
        {DBLP:conf/scn/Dinsdale-YoungM20}
\bibfield{author}{\bibinfo{person}{Thomas Dinsdale{-}Young},
  \bibinfo{person}{Bernardo Magri}, \bibinfo{person}{Christian Matt},
  \bibinfo{person}{Jesper~Buus Nielsen}, {and} \bibinfo{person}{Daniel
  Tschudi}.} \bibinfo{year}{2020}\natexlab{}.
\newblock \showarticletitle{Afgjort: {A} Partially Synchronous Finality Layer
  for Blockchains}. In \bibinfo{booktitle}{\emph{Security and Cryptography for
  Networks - 12th International Conference, {SCN} 2020, Amalfi, Italy,
  September 14-16, 2020, Proceedings}} \emph{(\bibinfo{series}{Lecture Notes in
  Computer Science}, Vol.~\bibinfo{volume}{12238})},
  \bibfield{editor}{\bibinfo{person}{Clemente Galdi} {and}
  \bibinfo{person}{Vladimir Kolesnikov}} (Eds.). \bibinfo{publisher}{Springer},
  \bibinfo{pages}{24--44}.
\newblock
\urldef\tempurl%
\url{https://doi.org/10.1007/978-3-030-57990-6\_2}
\showDOI{\tempurl}


\bibitem[Garay et~al\mbox{.}(2015)]%
        {EC:GarKiaLeo15}
\bibfield{author}{\bibinfo{person}{Juan~A. Garay}, \bibinfo{person}{Aggelos
  Kiayias}, {and} \bibinfo{person}{Nikos Leonardos}.}
  \bibinfo{year}{2015}\natexlab{}.
\newblock \showarticletitle{The Bitcoin Backbone Protocol: Analysis and
  Applications}. In \bibinfo{booktitle}{\emph{Advances in Cryptology -
  {EUROCRYPT} 2015 - 34th Annual International Conference on the Theory and
  Applications of Cryptographic Techniques, Sofia, Bulgaria, April 26-30, 2015,
  Proceedings, Part {II}}} \emph{(\bibinfo{series}{Lecture Notes in Computer
  Science}, Vol.~\bibinfo{volume}{9057})},
  \bibfield{editor}{\bibinfo{person}{Elisabeth Oswald} {and}
  \bibinfo{person}{Marc Fischlin}} (Eds.). \bibinfo{publisher}{Springer},
  \bibinfo{pages}{281--310}.
\newblock
\urldef\tempurl%
\url{https://doi.org/10.1007/978-3-662-46803-6\_10}
\showDOI{\tempurl}


\bibitem[Gilad et~al\mbox{.}(2017)]%
        {DBLP:conf/sosp/GiladHMVZ17}
\bibfield{author}{\bibinfo{person}{Yossi Gilad}, \bibinfo{person}{Rotem Hemo},
  \bibinfo{person}{Silvio Micali}, \bibinfo{person}{Georgios Vlachos}, {and}
  \bibinfo{person}{Nickolai Zeldovich}.} \bibinfo{year}{2017}\natexlab{}.
\newblock \showarticletitle{Algorand: Scaling Byzantine Agreements for
  Cryptocurrencies}. In \bibinfo{booktitle}{\emph{Proceedings of the 26th
  Symposium on Operating Systems Principles, Shanghai, China, October 28-31,
  2017}}. \bibinfo{publisher}{{ACM}}, \bibinfo{pages}{51--68}.
\newblock
\urldef\tempurl%
\url{https://doi.org/10.1145/3132747.3132757}
\showDOI{\tempurl}


\bibitem[Gilson(2013)]%
        {feathercoin-checkpoints}
\bibfield{author}{\bibinfo{person}{David Gilson}.}
  \bibinfo{year}{2013}\natexlab{}.
\newblock \bibinfo{title}{Feathercoin secures its block chain with advanced
  checkpointing}.
\newblock
\newblock
\newblock
\shownote{\url{https://www.coindesk.com/feathercoin-secures-block-chain-advanced-check-pointing}}.


\bibitem[Karakostas and Kiayias(2020)]%
        {cryptoeprint:2020:173}
\bibfield{author}{\bibinfo{person}{Dimitris Karakostas} {and}
  \bibinfo{person}{Aggelos Kiayias}.} \bibinfo{year}{2020}\natexlab{}.
\newblock \bibinfo{title}{Securing Proof-of-Work Ledgers via Checkpointing}.
\newblock \bibinfo{howpublished}{Cryptology ePrint Archive, Report 2020/173}.
\newblock
\newblock
\shownote{\url{https://eprint.iacr.org/2020/173}}.


\bibitem[Kiayias and Panagiotakos(2017)]%
        {LC:KiaPan17}
\bibfield{author}{\bibinfo{person}{Aggelos Kiayias} {and}
  \bibinfo{person}{Giorgos Panagiotakos}.} \bibinfo{year}{2017}\natexlab{}.
\newblock \showarticletitle{On Trees, Chains and Fast Transactions in the
  Blockchain}. In \bibinfo{booktitle}{\emph{Progress in Cryptology -
  {LATINCRYPT} 2017 - 5th International Conference on Cryptology and
  Information Security in Latin America, Havana, Cuba, September 20-22, 2017,
  Revised Selected Papers}} \emph{(\bibinfo{series}{Lecture Notes in Computer
  Science}, Vol.~\bibinfo{volume}{11368})},
  \bibfield{editor}{\bibinfo{person}{Tanja Lange} {and} \bibinfo{person}{Orr
  Dunkelman}} (Eds.). \bibinfo{publisher}{Springer}, \bibinfo{pages}{327--351}.
\newblock
\urldef\tempurl%
\url{https://doi.org/10.1007/978-3-030-25283-0\_18}
\showDOI{\tempurl}


\bibitem[Kiayias et~al\mbox{.}(2017)]%
        {C:KRDO17}
\bibfield{author}{\bibinfo{person}{Aggelos Kiayias}, \bibinfo{person}{Alexander
  Russell}, \bibinfo{person}{Bernardo David}, {and} \bibinfo{person}{Roman
  Oliynykov}.} \bibinfo{year}{2017}\natexlab{}.
\newblock \showarticletitle{Ouroboros: {A} Provably Secure Proof-of-Stake
  Blockchain Protocol}. In \bibinfo{booktitle}{\emph{Advances in Cryptology -
  {CRYPTO} 2017 - 37th Annual International Cryptology Conference, Santa
  Barbara, CA, USA, August 20-24, 2017, Proceedings, Part {I}}}
  \emph{(\bibinfo{series}{Lecture Notes in Computer Science},
  Vol.~\bibinfo{volume}{10401})}, \bibfield{editor}{\bibinfo{person}{Jonathan
  Katz} {and} \bibinfo{person}{Hovav Shacham}} (Eds.).
  \bibinfo{publisher}{Springer}, \bibinfo{pages}{357--388}.
\newblock
\urldef\tempurl%
\url{https://doi.org/10.1007/978-3-319-63688-7\_12}
\showDOI{\tempurl}


\bibitem[King and Nadal(2012)]%
        {king2012ppcoin}
\bibfield{author}{\bibinfo{person}{Sunny King} {and} \bibinfo{person}{Scott
  Nadal}.} \bibinfo{year}{2012}\natexlab{}.
\newblock \bibinfo{title}{Ppcoin: Peer-to-peer crypto-currency with
  proof-of-stake}.
\newblock
\newblock


\bibitem[Kokoris{-}Kogias et~al\mbox{.}(2016)]%
        {DBLP:conf/uss/Kokoris-KogiasJ16}
\bibfield{author}{\bibinfo{person}{Eleftherios Kokoris{-}Kogias},
  \bibinfo{person}{Philipp Jovanovic}, \bibinfo{person}{Nicolas Gailly},
  \bibinfo{person}{Ismail Khoffi}, \bibinfo{person}{Linus Gasser}, {and}
  \bibinfo{person}{Bryan Ford}.} \bibinfo{year}{2016}\natexlab{}.
\newblock \showarticletitle{Enhancing Bitcoin Security and Performance with
  Strong Consistency via Collective Signing}. In \bibinfo{booktitle}{\emph{25th
  {USENIX} Security Symposium, {USENIX} Security 16, Austin, TX, USA, August
  10-12, 2016}}, \bibfield{editor}{\bibinfo{person}{Thorsten Holz} {and}
  \bibinfo{person}{Stefan Savage}} (Eds.). \bibinfo{publisher}{{USENIX}
  Association}, \bibinfo{pages}{279--296}.
\newblock
\urldef\tempurl%
\url{https://www.usenix.org/conference/usenixsecurity16/technical-sessions/presentation/kogias}
\showURL{%
\tempurl}


\bibitem[Lerner(2015)]%
        {rsk}
\bibfield{author}{\bibinfo{person}{Sergio~Demian Lerner}.}
  \bibinfo{year}{2015}\natexlab{}.
\newblock \bibinfo{title}{RSK White paper Overview}.
\newblock
\newblock
\newblock
\shownote{\url{https://docs.rsk.co/RSK_White_Paper-Overview.pdf}}.


\bibitem[Lewenberg et~al\mbox{.}(2015)]%
        {inclusive}
\bibfield{author}{\bibinfo{person}{Yoad Lewenberg}, \bibinfo{person}{Yonatan
  Sompolinsky}, {and} \bibinfo{person}{Aviv Zohar}.}
  \bibinfo{year}{2015}\natexlab{}.
\newblock \showarticletitle{Inclusive block chain protocols}. In
  \bibinfo{booktitle}{\emph{International Conference on Financial Cryptography
  and Data Security}}. Springer, \bibinfo{pages}{528--547}.
\newblock


\bibitem[Li et~al\mbox{.}(2018)]%
        {conflux}
\bibfield{author}{\bibinfo{person}{Chenxing Li}, \bibinfo{person}{Peilun Li},
  \bibinfo{person}{Wei Xu}, \bibinfo{person}{Fan Long}, {and}
  \bibinfo{person}{Andrew Chi-chih Yao}.} \bibinfo{year}{2018}\natexlab{}.
\newblock \showarticletitle{Scaling Nakamoto Consensus to Thousands of
  Transactions per Second}.
\newblock \bibinfo{journal}{\emph{arXiv preprint arXiv:1805.03870}}
  (\bibinfo{year}{2018}).
\newblock


\bibitem[Nakamoto(2008)]%
        {nakamoto2008bitcoin}
\bibfield{author}{\bibinfo{person}{Satoshi Nakamoto}.}
  \bibinfo{year}{2008}\natexlab{}.
\newblock \bibinfo{title}{Bitcoin: A peer-to-peer electronic cash system}.
\newblock
\newblock


\bibitem[Nakamoto(2010)]%
        {btc-checkpoints}
\bibfield{author}{\bibinfo{person}{Satoshi Nakamoto}.}
  \bibinfo{year}{2010}\natexlab{}.
\newblock \bibinfo{title}{Bitcoin 0.3.2 released}.
\newblock
\newblock
\newblock
\shownote{\url{https://bitcointalk.org/index.php?topic=437.msg3807}}.


\bibitem[Neu et~al\mbox{.}(2020)]%
        {ebbandflow}
\bibfield{author}{\bibinfo{person}{Joachim Neu}, \bibinfo{person}{Ertem~Nusret
  Tas}, {and} \bibinfo{person}{David Tse}.} \bibinfo{year}{2020}\natexlab{}.
\newblock \bibinfo{title}{Ebb-and-Flow Protocols: A Resolution of the
  Availability-Finality Dilemma}.
\newblock \bibinfo{howpublished}{Cryptology ePrint Archive, Report 2020/1091}.
\newblock
\newblock
\shownote{\url{https://eprint.iacr.org/2020/1091}}.


\bibitem[Pass and Shi(2017a)]%
        {fruitchains}
\bibfield{author}{\bibinfo{person}{R. Pass} {and} \bibinfo{person}{E. Shi}.}
  \bibinfo{year}{2017}\natexlab{a}.
\newblock \showarticletitle{Fruitchains: A fair blockchain}. In
  \bibinfo{booktitle}{\emph{Proceedings of the ACM Symposium on Principles of
  Distributed Computing}}. ACM.
\newblock


\bibitem[Pass and Shi(2017b)]%
        {DBLP:conf/wdag/PassS17}
\bibfield{author}{\bibinfo{person}{Rafael Pass} {and} \bibinfo{person}{Elaine
  Shi}.} \bibinfo{year}{2017}\natexlab{b}.
\newblock \showarticletitle{Hybrid Consensus: Efficient Consensus in the
  Permissionless Model}. In \bibinfo{booktitle}{\emph{31st International
  Symposium on Distributed Computing, {DISC} 2017, October 16-20, 2017, Vienna,
  Austria}} \emph{(\bibinfo{series}{LIPIcs}, Vol.~\bibinfo{volume}{91})},
  \bibfield{editor}{\bibinfo{person}{Andr{\'{e}}a~W. Richa}} (Ed.).
  \bibinfo{publisher}{Schloss Dagstuhl - Leibniz-Zentrum f{\"{u}}r Informatik},
  \bibinfo{pages}{39:1--39:16}.
\newblock
\urldef\tempurl%
\url{https://doi.org/10.4230/LIPIcs.DISC.2017.39}
\showDOI{\tempurl}


\bibitem[Pass and Shi(2018)]%
        {DBLP:conf/eurocrypt/PassS18}
\bibfield{author}{\bibinfo{person}{Rafael Pass} {and} \bibinfo{person}{Elaine
  Shi}.} \bibinfo{year}{2018}\natexlab{}.
\newblock \showarticletitle{Thunderella: Blockchains with Optimistic Instant
  Confirmation}. In \bibinfo{booktitle}{\emph{Advances in Cryptology -
  {EUROCRYPT} 2018 - 37th Annual International Conference on the Theory and
  Applications of Cryptographic Techniques, Tel Aviv, Israel, April 29 - May 3,
  2018 Proceedings, Part {II}}} \emph{(\bibinfo{series}{Lecture Notes in
  Computer Science}, Vol.~\bibinfo{volume}{10821})},
  \bibfield{editor}{\bibinfo{person}{Jesper~Buus Nielsen} {and}
  \bibinfo{person}{Vincent Rijmen}} (Eds.). \bibinfo{publisher}{Springer},
  \bibinfo{pages}{3--33}.
\newblock
\urldef\tempurl%
\url{https://doi.org/10.1007/978-3-319-78375-8\_1}
\showDOI{\tempurl}


\bibitem[Research(2020)]%
        {eth2}
\bibfield{author}{\bibinfo{person}{Ethereum Research}.}
  \bibinfo{year}{2020}\natexlab{}.
\newblock \bibinfo{title}{Ethereum 2.0}.
\newblock
  \bibinfo{howpublished}{\url{https://github.com/ethereum/eth2.0-specs}}.
\newblock


\bibitem[Sankagiri et~al\mbox{.}(2020)]%
        {sankagiri2020checkpointed}
\bibfield{author}{\bibinfo{person}{Suryanarayana Sankagiri},
  \bibinfo{person}{Xuechao Wang}, \bibinfo{person}{Sreeram Kannan}, {and}
  \bibinfo{person}{Pramod Viswanath}.} \bibinfo{year}{2020}\natexlab{}.
\newblock \bibinfo{title}{The Checkpointed Longest Chain: User-dependent
  Adaptivity and Finality}.
\newblock
\newblock
\showeprint[arxiv]{2010.13711}~[cs.CR]


\bibitem[Sheng et~al\mbox{.}(2020)]%
        {BFTforensics}
\bibfield{author}{\bibinfo{person}{Peiyao Sheng}, \bibinfo{person}{Gerui Wang},
  \bibinfo{person}{Kartik Nayak}, \bibinfo{person}{Sreeram Kannan}, {and}
  \bibinfo{person}{Pramod Viswanath}.} \bibinfo{year}{2020}\natexlab{}.
\newblock \bibinfo{title}{BFT Protocol Forensics}.
\newblock
\newblock
\showeprint[arxiv]{2010.06785}~[cs.CR]


\bibitem[Stewart and Kokoris{-}Kogia(2020)]%
        {DBLP:journals/corr/abs-2007-01560}
\bibfield{author}{\bibinfo{person}{Alistair Stewart} {and}
  \bibinfo{person}{Eleftherios Kokoris{-}Kogia}.}
  \bibinfo{year}{2020}\natexlab{}.
\newblock \showarticletitle{{GRANDPA:} a Byzantine Finality Gadget}.
\newblock \bibinfo{journal}{\emph{CoRR}}  \bibinfo{volume}{abs/2007.01560}
  (\bibinfo{year}{2020}).
\newblock
\showeprint[arxiv]{2007.01560}
\urldef\tempurl%
\url{https://arxiv.org/abs/2007.01560}
\showURL{%
\tempurl}


\bibitem[Yang et~al\mbox{.}(2019)]%
        {prism_systems}
\bibfield{author}{\bibinfo{person}{L. Yang}, \bibinfo{person}{V. Bagaria},
  \bibinfo{person}{Gerui Wang}, \bibinfo{person}{Mohammad Alizadeh},
  \bibinfo{person}{David~M. Tse}, \bibinfo{person}{G. Fanti}, {and}
  \bibinfo{person}{P. Viswanath}.} \bibinfo{year}{2019}\natexlab{}.
\newblock \showarticletitle{Prism: Scaling Bitcoin by 10,000x}.
\newblock \bibinfo{journal}{\emph{ArXiv}}  \bibinfo{volume}{abs/1909.11261}
  (\bibinfo{year}{2019}).
\newblock


\bibitem[Yu et~al\mbox{.}(2019)]%
        {yu2019ohie}
\bibfield{author}{\bibinfo{person}{Haifeng Yu}, \bibinfo{person}{Ivica
  Nikolic}, \bibinfo{person}{Ruomu Hou}, {and} \bibinfo{person}{Prateek
  Saxena}.} \bibinfo{year}{2019}\natexlab{}.
\newblock \bibinfo{title}{OHIE: Blockchain Scaling Made Simple}.
\newblock
\newblock
\showeprint[arxiv]{1811.12628}~[cs.DC]


\end{thebibliography}

\appendix

%
%
%
%
\end{document}